%% file: main.tex
\documentclass[a4paper,envcountsame,10pt]{article}
\usepackage{authblk}

\title{The Maximal MAM,\\ a Reasonable Implementation of\\ the Maximal Strategy}
\author{Beniamino Accattoli}
\affil{INRIA, UMR 7161, LIX, \'Ecole Polytechnique, France}

\input{\macrospath/ben-macros-preamble}
\setboolean{needstheorems}{true}
\input{\macrospath/ben-macros-paper}

\newtheorem*{theorem*}{Theorem}

\input{local-macros}
\date{}
\begin{document}
\maketitle


\input{00_-_Introduction}
\input{01_-_Lambda-Calculus_and_Maximal_Evaluation}

\input{02_-_Preliminaries_on_Abstract_Machines}
\input{03_-_The_Checking_Abstract_Machine}
\input{04_-_The_Maximal_Milner_Abstract_Machine}

\input{05_-_Quantitative_Analysis}

\bibliographystyle{alpha}

\withproofs{\input{99_-_appendix-minor}}

\end{document}

%% file: ben-macros-preamble.tex
\usepackage{ifthen}

\newboolean{talk}
\setboolean{talk}{false}
\newboolean{paper}
\setboolean{paper}{false}

\newboolean{IEEEstyle}
\setboolean{IEEEstyle}{false}
\newboolean{lipicsstyle}
\setboolean{lipicsstyle}{false}
\newboolean{eptcsstyle}
\setboolean{eptcsstyle}{false}
\newboolean{entcsstyle}
\setboolean{entcsstyle}{false}
\newboolean{sigplanstyle}
\setboolean{sigplanstyle}{false}
\newboolean{easychairstyle}
\setboolean{easychairstyle}{false}

\newboolean{needstheorems}
\setboolean{needstheorems}{false}
\newboolean{withimages}
\setboolean{withimages}{false}
\newboolean{withproofs}
\setboolean{withproofs}{true}

\newboolean{french}
\setboolean{french}{false}

%% file: ben-macros-paper.tex
\ifthenelse{\boolean{talk}}{}{\usepackage[utf8]{inputenc}}
\ifthenelse{\boolean{french}}{\usepackage[french]{babel}}{
  \ifthenelse{\boolean{lipicsstyle}}{}{\usepackage[english]{babel}}}
\usepackage{amssymb}
\usepackage{amsmath}
\usepackage{graphicx}
\usepackage{bussproofs}
\usepackage{cmll}
\usepackage{stmaryrd}
\usepackage{multirow}
\ifthenelse{\boolean{talk}}{\usepackage{color}}{
  \ifthenelse{\boolean{lipicsstyle}}{\usepackage{color}}{\usepackage[usenames]{xcolor}}}
\ifthenelse{\boolean{IEEEstyle}}{
\usepackage[bookmarks={false}]{hyperref}}{
\usepackage{hyperref}}
\usepackage{fancybox}
\usepackage{xspace}
\usepackage{hhline}

\ifthenelse{\boolean{IEEEstyle}}{
\setboolean{needstheorems}{true}}{}

\ifthenelse{\boolean{eptcsstyle}}{
\setboolean{needstheorems}{true}}{}

\ifthenelse{\boolean{sigplanstyle}}{
\setboolean{needstheorems}{true}}{}

\ifthenelse{\boolean{easychairstyle}}{
\setboolean{needstheorems}{true}}{}

\ifthenelse{\boolean{lipicsstyle}}{
    \newtheorem{proposition}[theorem]{Proposition}
    }{}

\ifthenelse{\boolean{needstheorems}}{
\input{\macrospath/ben-macros-thm-environments}}{}

\newcommand{\myproof}[1]{
\ifthenelse{\boolean{withproofs}}{#1}{}
}

\newcommand{\la}[1]{\lambda #1.}

\newcommand{\tm}{t}
\newcommand{\tmtwo}{s}
\newcommand{\tmthree}{u}
\newcommand{\tmfour}{r}
\newcommand{\tmfive}{p}
\newcommand{\tmsix}{q}

\newcommand{\var}{x}
\newcommand{\vartwo}{y}
\newcommand{\varthree}{z}

\newcommand{\rootRew}[1]{\mapsto_{#1}}
\newcommand{\Rew}[1]{\rightarrow_{#1}}

\renewcommand{\to}{\Rew{}}
\newcommand{\tob}{\Rew{\beta}}
\newcommand{\rtob}{\rootRew{\beta}}

\newcommand{\osym}{o}

\newcommand{\esym}{{\mathtt e}}

\newcommand{\msym}{{\mathtt m}}

\newcommand{\csym}{{\tt c}}

\newcommand{\ctxholep}[1]{\langle #1\rangle}
\newcommand{\ctxhole}{\ctxholep{\cdot}}

\newcommand{\ctx}{C}
\newcommand{\ctxtwo}{D}
\newcommand{\ctxthree}{E}

\newcommand{\ctxp}[1]{\ctx\ctxholep{#1}}

\newcommand{\ctxtwop}[1]{\ctxtwo\ctxholep{#1}}

\newcommand{\nbvctxtwo}[1]{\nbvctxtwo{#1}}

\newcommand{\defeq}{:=}
\newcommand{\grameq}{::=}

\newcommand{\isub}[2]{\{#1/#2\}}
\newcommand{\esub}[2]{[#1/#2]}
\renewcommand{\esub}[2]{[#1{\shortleftarrow}#2]}
\renewcommand{\isub}[2]{\{#1{\shortleftarrow}#2\}}

\newcommand{\llbrace}{\{ \kern -0.27em \vert}
\newcommand{\rrbrace}{\vert \kern -0.27em \}}

\renewcommand{\l}{\lambda}
\newcommand{\ie}{{\em i.e.}\xspace}
\newcommand{\eg}{{\em e.g.}\xspace}
\newcommand{\ih}{{\textit{i.h.}}\xspace}
\newcommand{\fv}[1]{{\tt fv}(#1)}

\newcommand{\red}[1]{{\color{red} {#1}}}

\newcommand{\ignore}[1]{}

\newcommand{\myinput}[1]{\ifthenelse{\boolean{withimages}}{\input{#1}}{}}

\newcommand{\reflemma}[1]{Lemma~\ref{l:#1}}
\newcommand{\reflemmap}[2]{Lemma~\ref{l:#1}.\ref{p:#1-#2}}
\newcommand{\reflemmaeq}[1]{{L.\ref{l:#1}}}
\newcommand{\reflemmaeqp}[2]{{L.\ref{l:#1}.\ref{p:#1-#2}}}

\newcommand{\refthm}[1]{Theorem~\ref{thm:#1}}
\newcommand{\refthmp}[2]{Theorem~\ref{thm:#1}.\ref{p:#1-#2}}
\newcommand{\refthmcompact}[1]{Thm~\ref{thm:#1}}

\newcommand{\refsect}[1]{Sect.~\ref{sect:#1}}

\ifthenelse{\boolean{easychairstyle}}{}{
	
}
\newcommand{\reffig}[1]{Fig.~\ref{fig:#1}}

\newcommand{\refdef}[1]{Definition~\ref{def:#1}}

\newcommand{\refrem}[1]{Remark~\ref{rem:#1}}
\newcommand{\refpoint}[1]{Point~\ref{p:#1}}
\newcommand{\refpointeq}[1]{P.\ref{p:#1}}

\newcommand{\set}[1]{\{#1\}}

\newcommand{\nat}{\mathbb{N}}

\newcommand{\size}[1]{|#1|}
\newcommand{\sizep}[2]{|#1|_{#2}}
\newcommand{\sizem}[1]{\sizep{#1}{\msym}}
\newcommand{\sizee}[1]{\sizep{#1}{\esym}}

\newcommand{\code}{\overline{\tm}}
\newcommand{\codetwo}{\overline{\tmtwo}}
\newcommand{\codethree}{\overline{\tmthree}}
\newcommand{\codefour}{\overline{\tmfour}}

\newcommand{\genv}{E}
\newcommand{\genvtwo}{\genv'}
\newcommand{\genvthree}{\genv''}
\newcommand{\genvfour}{\genv'''}

\newcommand{\stack}{\pi}
\newcommand{\stacktwo}{\stack'}
\newcommand{\stackthree}{\stack''}

\newcommand{\state}{s}
\newcommand{\statetwo}{\state'}
\newcommand{\statethree}{\state''}
\newcommand{\statefour}{\state'''}

\newcommand{\exec}{\rho}
\newcommand{\exectwo}{\sigma}
\newcommand{\execthree}{\tau}

\newcommand{\decode}[1]{\underline{#1}}
\newcommand{\decodep}[2]{\decode{#1}\ctxholep{#2}}
\newcommand{\decoderevp}[2]{\ctxholep{#2}\decode{#1}}

\newcommand{\stempty}{\epsilon}
\newcommand{\cons}{::}

\newcommand{\deriv}{\ensuremath{d}}

\newcommand{\derivtwo}{\ensuremath{e}}

\newcommand{\unfsym}{\rotatebox[origin=c]{-90}{$\rightarrow$}}
\newcommand{\unf}[1]{#1\unfsym\,}
\newcommand{\relunf}[2]{\unf{#1}_{#2}}

\newcommand{\rename}[1]{#1^\alpha}

\newcommand{\stackitem}{\phi}

\newcommand{\lab}{\red l}

\renewcommand{\dump}{D}

\newcommand{\dentry}[2]{#1\lozenge#2}
\newcommand{\gcentry}[2]{#1 \blacklozenge#2 }

\newcommand{\decstack}{\decode{\pi}}

\newcommand{\decstackp}[1]{\decstack\ctxholep{#1}}

\newcommand{\mach}{{\tt M}}

\newcommand{\admsym}{\csym}

\newcommand{\tomachhole}[1]{\leadsto_{#1}}
\newcommand{\tomach}{\tomachhole{}}
\newcommand{\tomachm}{\tomachhole{\msym}}
\newcommand{\tomachmone}{\tomachhole{\monesym}}
\newcommand{\tomachmtwo}{\tomachhole{\mtwosym}}
\newcommand{\tomachmthree}{\tomachhole{\mthreesym}}
\newcommand{\tomache}{\tomachhole{\esym}}
\newcommand{\tomachc}{\tomachhole{\admsym}}

\newcommand{\midd}{\mid}

\newcommand{\withproofs}[1]{\ifthenelse{\boolean{withproofs}}{#1}{}}

\newcommand{\withoutproofs}[1]{\ifthenelse{\boolean{withproofs}}{}{#1}}

\newcommand{\maca}[5]{#2 & #3 & #4 & #5 & #1}
\newcommand{\skamstate}[5]{#2\mid#3\mid#4\mid#5\mid#1}

\newcommand{\skphase}{\varphi}
\newcommand{\skeval}{\textcolor{blue}{\blacktriangledown}}
\newcommand{\skback}{\textcolor{red}{\blacktriangle}}
\newcommand{\mframe}{\skframe}
\newcommand{\skframe}{F}

\newcommand{\skframetwo}{\skframe'}
\newcommand{\skframethree}{\skframe''}

\newcommand{\lababs}{abs}
\newcommand{\labredsym}{red}
\newcommand{\labred}[1]{(\labredsym,#1)}
\newcommand{\labneu}{neu}

\newcommand{\commone}{\skeval\csym_1}
\newcommand{\commtwo}{\skeval\csym_2}
\newcommand{\commthree}{\skeval\csym_3}
\newcommand{\commfour}{\skback\csym_4}
\newcommand{\commfive}{\skback\csym_5}
\newcommand{\commsix}{\skback\csym_6}
\newcommand{\commseven}{\skeval\csym_7}

\newcommand{\eabssym}{\esym_{\lababs}}
\newcommand{\eredsym}{\esym_{\labredsym}}

\newcommand{\monesym}{\msym_1}
\newcommand{\mtwosym}{\msym_2}
\newcommand{\mthreesym}{\msym_3}

\newcommand{\tomachered}{\tomachhole{\eredsym}}
\newcommand{\tomacheabs}{\tomachhole{\eabssym}}
\newcommand{\tomacheredp}[1]{\tomachhole{\esym_{\labred {#1}}}}

\newcommand{\tomachasubone}{\tomachhole\commone}
\newcommand{\tomachasubtwo}{\tomachhole\commtwo}
\newcommand{\tomachasubthree}{\tomachhole\commthree}
\newcommand{\tomachasubfour}{\tomachhole\commfour}
\newcommand{\tomachasubfive}{\tomachhole\commfive}
\newcommand{\tomachasubsix}{\tomachhole\commsix}
\newcommand{\tomachasubseven}{\tomachhole{\commseven}}

\newcommand{\esublab}[2]{[#1{\shortleftarrow}#2]^{\lab}}
\newcommand{\esublabe}[3]{[#1{\shortleftarrow}#2]^{#3}}

\newcommand{\ucam}{Checking AM\xspace}
\newcommand{\ecnf}[1]{\mathtt{nf}_{\esym\admsym}(#1)}

\newcommand{\skap}[2]{\dentry{#1}{#2}}

\newcommand{\stctx}[1]{\ctx_{#1}}
\newcommand{\stctxp}[2]{\ctx_{#1}\ctxholep{#2}}

\newcommand{\sizex}[1]{\size{#1}_{\mathtt{x}}}

\newcommand{\decodefun}{\decode{{ }\cdot{ }}}
\newcommand{\csub}[2]{\llbrace#1{\shortleftarrow}#2\rrbrace}

\newcommand{\toucmach}{\rightharpoonup}
\newcommand{\toucmachhole}[1]{\toucmach_{#1}}
\newcommand{\toucmachasubone}{\toucmachhole\commone}
\newcommand{\toucmachasubtwo}{\toucmachhole\commtwo}
\newcommand{\toucmachasubthree}{\toucmachhole\commthree}
\newcommand{\toucmachasubfour}{\toucmachhole\commfour}
\newcommand{\toucmachasubfive}{\toucmachhole\commfive}
\newcommand{\toucmachasubsix}{\toucmachhole\commsix}
\newcommand{\toucmachasubseven}{\toucmachhole\commseven}

\newcommand{\toucmachoutone}{\toucmachhole{\osym_1}}
\newcommand{\toucmachouttwo}{\toucmachhole{\osym_2}}
\newcommand{\toucmachoutthree}{\toucmachhole{\osym_3}}
\newcommand{\toucmachoutfour}{\toucmachhole{\osym_4}}
\newcommand{\toucmachoutfive}{\toucmachhole{\osym_5}}
\newcommand{\toucmachoutsix}{\toucmachhole{\osym_6}}

\newcommand{\decdump}{\decode \dump}
\newcommand{\decdumpp}[1]{\decdump \ctxholep{#1}}

\newcommand{\ucamstate}[4]{(#1,#2,#3,#4)}
\newcommand{\mac}[5]{(#2,#3,#4,#5,#1)}

\newcommand{\decodestack}[2]{\ctxholep{#1}\decode{#2}}

\newcommand{\sizecom}[1]{|#1|_c}

\newcommand{\sizecomev}[1]{|#1|_{\skeval \csym}}
\newcommand{\sizecombt}[1]{|#1|_{\skback \csym}}

 \ifthenelse{\boolean{entcsstyle}}{}{
 \newcommand{\case}[1]{{\bf Case #1.}}
 }

\newcommand{\rmeasure}[1]{\size{#1}}

\newcommand{\polsize}[2]{|#1|_{#2}}

\newcommand{\quiet}{neutral\xspace}

\newcommand{\closescopem}[1]{\skback#1}

\renewcommand{\perp}{Max\xspace}
\newcommand{\perpbeta}{{\tiny \mathtt{\perp}}\beta}
\newcommand{\toperp}{\Rew{\perpbeta}}
\newcommand{\rtobperp}{\rootRew{\perpbeta}}

\newcommand{\perpRuleAx}{(ax)}
\newcommand{\perpRuleLambda}{($\l$)}
\newcommand{\perpRuleAppR}{(@r)}
\newcommand{\perpRuleAppL}{(@l)}
\newcommand{\perpRuleGc}{(gc)}

\newcommand{\maxmam}{Max MAM\xspace}

\newenvironment{varitemize}
{
\begin{list}{\labelitemi}
{
\setlength{\itemsep}{.7pt}
 \setlength{\topsep}{.7pt}
 \setlength{\parsep}{.7pt}
 \setlength{\partopsep}{.7pt}
 \setlength{\leftmargin}{15pt}
 \setlength{\rightmargin}{0pt}
 \setlength{\itemindent}{0pt}
 \setlength{\labelsep}{5pt}
 \setlength{\labelwidth}{15pt}
}}
{
 \end{list} 
}

\newcounter{numberone}

\newenvironment{varenumerate}
{
\begin{list}{\arabic{numberone}.}
{
 \usecounter{numberone}
 \setlength{\itemsep}{.7pt}
 \setlength{\topsep}{.7pt}
 \setlength{\parsep}{.7pt}
 \setlength{\partopsep}{.7pt}
 \setlength{\leftmargin}{15pt}
 \setlength{\rightmargin}{0pt}
 \setlength{\itemindent}{0pt}
 \setlength{\labelsep}{5pt}
 \setlength{\labelwidth}{15pt}
}}
{
\end{list} 
}

%% file: ben-macros-thm-environments.tex
\usepackage{amsthm}

    \newtheorem{theorem}{Theorem}[section]
    \newtheorem{lemma}[theorem]{Lemma}
    \newtheorem{definition}[theorem]{Definition}
    \newtheorem{remark}[theorem]{Remark}

%% file: local-macros.tex
\newcommand\scalemath[2]{\scalebox{#1}{\mbox{\ensuremath{\displaystyle #2}}}}
\usepackage[font=small,skip=0pt]{caption}
\usepackage{hhline}

\renewcommand{\tmtwo}{u}
\renewcommand{\tmthree}{r}
\renewcommand{\tmfour}{s}

%% file: 00_-_Introduction.tex

\section{Introduction}

This note is about a reasonable abstract machine, called Maximal MAM, implementing the maximal strategy of the $\l$-calculus, that is, the strategy that always produces a longest evaluation sequence. The abstract machine is a minor variation over the Useful MAM of \cite{DBLP:conf/wollic/Accattoli16}, that is a reasonable implementation of the leftmost-outermost strategy. Here \emph{reasonable} is a technical term: an abstract machine $\mach$ implementing a strategy $\to$ is reasonable when its overhead on a given term $\tm$ is polynomial with respect to the size of $\tm$ and the number of $\to$-steps necessary to evaluate $\tm$.

The result of this note is discussed and put in context in the author's \emph{(In)Efficiency and Reasonable Cost Models} \cite{LSFA2017invited}. Essentially, the maximal strategy is a peculiar case-study, for the following reasons:
\begin{itemize}
	\item \emph{Non-trivial reasonable implementation}: it is the most inefficient strategy ever, and yet implementing it with a reasonable overhead requires the whole heavy machinery, namely \emph{useful sharing}, needed for the leftmost-outermost strategy, the other only strong strategy known to be reasonable.
	\item \emph{Unreasonable cost model}: it is reasonably implementable---that usually is the tricky part in the study of cost models---but its number of steps does not provide a reasonable cost model. Indeed the strategy is so inefficient that it cannot even simulate Turing machines, that usually is the easy part in the study of cost models. To be precise, there are no proofs of such impossibility, but the known encodings of Turing machines do not work when the evaluation strategy in the $\l$-calculus is the maximal one.
\end{itemize}

We first define the strategy, and then show how to implement it via the \maxmam. The technical details follow closely those in \cite{DBLP:conf/wollic/Accattoli16} for the Useful MAM. In turn, the Useful MAM is a refinement of the Strong MAM, an unreasonable abstract machine for the leftmost-outermost strategy studied in \cite{DBLP:conf/aplas/AccattoliBM15}, that in turn is a simplification of Cregut's machine \cite{DBLP:conf/lfp/Cregut90,DBLP:journals/lisp/Cregut07,DBLP:conf/ppdp/Garcia-PerezNM13}. The literature about abstract machines for strong evaluation is extremely limited, we essentially already cited all existing papers on the subject. A few further papers \cite{DBLP:conf/icfp/GregoireL02,DBLP:conf/lics/AccattoliC15,DBLP:conf/fsen/AccattoliG17} deal with abstract machines for weak evaluation with open terms, that is an intermediate framework between the one of weak evaluation with closed terms, to which the almost totality of the literature is devoted, and the almost inexistent one of strong evaluation.

%% file: 01_-_Lambda-Calculus_and_Maximal_Evaluation.tex

\section{\texorpdfstring{$\l$}{lambda}-Calculus and Maximal Evaluation}

The syntax of the ordinary $\l$-calculus is given by the following grammar for
terms:
\[\begin{array}{c@{\hspace{2em}}ccccc}
   \mbox{\textsc{$\l$-Terms}} & \tm,\tmtwo,\tmthree,\tmfour\grameq\var\midd \l \var. \tm \midd \tm \tmtwo.
  \end{array}\]
We use $\tm\isub{\var}{\tmtwo}$ for the usual (meta-level)
notion of substitution. An abstraction $\la
\var \tm$ binds $\var$ in $\tm$, and we
silently work modulo $\alpha$-equivalence of  bound variables,
\eg\ $(\la\vartwo(\var\vartwo))\isub\var\vartwo =
\la\varthree(\vartwo\varthree)$. We use $\fv\tm$ for the set of
free variables of $\tm$.

\emph{$\beta$-reduction.} We define $\beta$-reduction $\tob$ as follows:
$$\begin{array}{ccc}
  \textsc{Contexts} \\   
	\ctx \grameq \ctxhole\midd \la \var \ctx\midd \ctx \tm \midd\tm\ctx \\\\
	\textsc{Rule at top level}\\
	(\l\var.\tm)\tmtwo \rtob \tm\isub\var\tmtwo \\\\
	\textsc{Contextual closure} \\
        \ctxp \tm \tob \ctxp \tmtwo \textrm{~~~if } \tm \rtob \tmtwo
\end{array}$$
 A context $\ctx$ is \emph{applicative} if $\ctx = \ctxtwop{\ctxhole\codetwo}$ for some $\ctxtwo$ and $\codetwo$. A term $\tm$ is a \emph{normal form}, or simply \emph{normal}, if there is no $\tmtwo$ such that $\tm\tob\tmtwo$, and it is \emph{neutral} if it is normal and it is not of the form $\la\var\tmtwo$ (\ie\ it is not an abstraction). The \emph{position} of a $\beta$-redex $\ctxp \tm \tob \ctxp \tmtwo$ is the context $\ctx$ in which it takes place. To ease the language, we will identify a redex with its position. A \emph{derivation} $\deriv:\tm\to^k\tmtwo$
is a finite, possibly empty, sequence of reduction steps. We write $\size\tm$ for the size of $\tm$ and $\size\deriv$ for the length of
$\deriv$.

\paragraph{Maximal Evaluation.} The maximal strategy is the variation over leftmost-outermost (LO) evaluation in which, when the LO redex $(\la\var\tm)\tmtwo$ is erasing (that is, when $\var \notin \fv\tm$), the strategy first evaluates (maximally) $\tmtwo$ to $\tmtwo'$, and then fires the erasing redex $(\la\var\tm)\tmtwo'\tob \tm$, to avoid the erasure of the $\beta$-redexes of $\tmtwo$ before being reduced. Of course, if $\tmtwo$ diverges than the maximal strategy diverges, while the LO strategy would not (or, at least, not because of $\tmtwo$). The LO strategy is inefficient but it has the key property of being \emph{normalizing}, \ie it reaches a normal form whenever it exists. The maximal strategy, dually, is \emph{perpetual}, that is, it diverges whenever possible. See van Raamsdonk at al.'s \cite{DBLP:journals/iandc/RaamsdonkSSX99} for more about perpetual and maximal strategies.

We define the maximal strategy by first defining the notion of \emph{maximal context}, that is a context in which a maximal redex can appear (and \emph{not} a context that cannot be extended).
 \begin{definition}[\perp Context]
 \label{def:ilob-ctx}
 Maximal (or max) contexts are defined by induction as follows:
 \begin{center}
$\begin{array}{c@{\hspace{1em}}c@{\hspace{1em}}cccc}
	\AxiomC{}
	\RightLabel{\perpRuleAx}
	\UnaryInfC{$\ctxhole$ is \perp}
	\DisplayProof
	&
	\AxiomC{$\ctx$ is \perp}
	\AxiomC{$\ctx\neq\la\var\ctxtwo$}
	\RightLabel{\perpRuleAppL}
	\BinaryInfC{$\ctx \tm$ is \perp}
	\DisplayProof 
 \\\\
	
	\AxiomC{$\ctx$ is \perp}
	\RightLabel{\perpRuleLambda}
	\UnaryInfC{$\la\var\ctx$ is \perp}
	\DisplayProof &
	\AxiomC{$\var \notin \fv\tm$}
	\AxiomC{$\ctx$ is \perp}	
	\RightLabel{\perpRuleGc}
	\BinaryInfC{$(\la\var\tm) \ctx$ is \perp}
	\DisplayProof
	\\\\
	&
	\AxiomC{$\tm$ is neutral}
	\AxiomC{$\ctx$ is \perp}
	\RightLabel{\perpRuleAppR}
	\BinaryInfC{$\tm \ctx$ is \perp}
	\DisplayProof 
\end{array}$
\end{center}
 \end{definition}

We define the \emph{maximal $\beta$-reduction strategy} $\toperp$ as follows:
$$\begin{array}{cc}
  \textsc{Rule at top level} \\
	(\la\var\tm)\tmtwo \rtobperp \tm\isub\var\tmtwo  \mbox{\ \ if $\var\in \fv\tm$ or $\tmtwo$ is normal}\\\\
	\textsc{Contextual closure} \\
        \ctxp \tm \toperp \ctxp \tmtwo \textrm{~~~if } \tm \rtobperp \tmtwo
\end{array}$$

As expected,

\begin{lemma}[Basic Properties of the Maximal Strategy, Proof at page \pageref{ssect:perp-properties}]
\label{l:perp-properties} 
Let $\tm$ be a $\l$-term that is not normal. Then
\begin{varenumerate}
	\item \emph{Completeness}: \label{p:perp-properties-completeness}
	there exists $\tmtwo$ such that $\tm\toperp\tmtwo$.
	\item \emph{Determinism}: \label{p:perp-properties-determinism}
	moreover, such a $\tmtwo$ is unique.
\end{varenumerate}
\end{lemma}

Let us point out that we talk about \emph{the} maximal strategy despite the fact that there are many maximal strategies, that can be obtained from our definition by adding some more freedom in the choice of the redex. A typical example is given by $\var \tm\tmtwo$: our strategy first reduces $\tm$ and then (if $\tm$ terminates) it evaluates $\tmtwo$, while one can interleave the evaluations of $\tm$ and $\tmtwo$ and still be maximal (because $\tm$ cannot act on $\tmtwo$ in $\var\tm\tmtwo$). We refer it as \emph{the} maximal strategy, because it is the leftmost maximal strategy, \ie it mostly behaves as the leftmost strategy, it only changes its behavior on erasing redexes (and the leftmost strategy is a sort of canonical strategy).

%% file: 02_-_Preliminaries_on_Abstract_Machines.tex
\section{Preliminaries on Abstract Machines}
We study two abstract machines, the Maximal Milner Abstract Machine (\maxmam) (\reffig{UMAM}) and an auxiliary machine called the \ucam (\reffig{UCAM}). The \maxmam is a reasonable implementation of the maximal strategy resting on labeled environments to implement useful sharing, and on the \ucam to produce these labels.

The \maxmam is meant to implement the maximal strategy via a decoding function $\decodefun$ mapping machine states to $\l$-terms. Machine states $\state$ are given by a \emph{code} $\code$, that is a $\l$-term $\tm$ \emph{not considered up to $\alpha$-equivalence} (which is why it is over-lined), and some data-structures like stacks, frames, and environments. The data-structures are used to implement the search for the next maximal redex and a form of micro-steps substitution, and they decode to evaluation contexts for $\toperp$. Every state $\state$ decodes to a term $\decode\state$, having the shape $\stctx\state\ctxholep\code$, where $\code$ is the code currently under evaluation and $\stctx\state$ is the evaluation context given by the data-structures. 

The \ucam uses the same states and data-structures of the \maxmam.\smallskip

\textbf{The Data-Structures.} First of all, our machines are executed on \emph{well-named} terms, that are those $\alpha$-representants where all variables (both bound and free) have distinct names. Then, the data-structures used by the machines are defined in \reffig{grammars}, namely:
\begin{varenumerate}
\item  \emph{Stack $\stack$}: it contains the arguments of the current code;

\item \emph{Frame $\mframe$}: a second stack, that together with  $\stack$ is used to walk through the term and search for the next redex to reduce. The items $\stackitem$ of a frame are of three kinds:
\begin{varenumerate}
\item \emph{Variables}: a variable $\var$ is pushed on the frame $\mframe$ whenever the machines starts evaluating under an abstraction $\l\var$. 
\item \emph{Head argument contexts}: $\dentry\code\stack$ is pushed on $\mframe$ every time evaluation enters in the right subterm $\codetwo$ of an application $\code\codetwo$. The entry saves the left part $\code$ of the application and the current stack $\stack$, to restore them when the evaluation of the right subterm $\codetwo$ is over. 
\item \emph{Erasing Contexts}: $\gcentry{ \la\var\code }\stack$ is pushed on $\mframe$ every time evaluation finds an erasing redex $(\la\var\code) \codetwo$ (that is, a redex for which $\var \notin \fv\code$). In this case evaluation enters in the argument $\codetwo$ to normalize it (or, possibly, to diverge) before erasing it.
\end{varenumerate}

\item \emph{Global Environment $\genv$}: it is used to implement micro-step substitution  (\ie on a variable occurrence at the time), storing the arguments of $\beta$-redexes that have been encountered so far. Most of the literature on abstract machines uses \emph{local environments} and \emph{closures}. Having just one global environment $\genv$ (used only in a minority of works  \cite{DBLP:journals/entcs/FernandezS09,DBLP:conf/birthday/SandsGM02,DBLP:conf/ppdp/DanvyZ13,DBLP:conf/icfp/AccattoliBM14,DBLP:conf/lics/AccattoliC15,DBLP:conf/aplas/AccattoliBM15,DBLP:conf/wollic/Accattoli16,DBLP:conf/fsen/AccattoliG17}, and discussed at length in \cite{DBLP:conf/ppdp/AccattoliB17}) removes the need for closures and simplifies the machine. On the other hand, it forces to use explicit $\alpha$-renamings (the operation $\rename\code$ in $\tomachered$ and $\tomacheabs$ in \reffig{UMAM}), but this does not affect the overall complexity, as it speeds up other operations, see \cite{DBLP:conf/ppdp/AccattoliB17}.
The entries of $\genv$ are of the form $\esublab\var\code$, \ie they carry a label $\lab$ used to implement usefulness, to be explained later on in this section. We write $\genv(\var) = \esublab\var\code$ when $\genv$ contains $\esublab\var\code$ and $\genv(\var) = \bot$ when in $\genv$ there are no entries of the form $\esublab\var\code$.
\end{varenumerate}\smallskip

\textbf{The Decoding.} Every state $\state$ decodes to a term $\decode\state$ (see \reffig{decoding}) of shape $\stctxp\state{\relunf\code\genv}$, where 
\begin{varenumerate}
 \item $\relunf\code\genv$ is a $\l$-term, roughly obtained by applying to the code the substitution induced by the global environment $\genv$. More precisely, the operation $\relunf\code\genv$ is called \emph{unfolding} and it is properly defined at the end of this section.
 \item $\stctx\state$ is a context, that will be shown to be a max context, obtained by decoding the stack $\stack$ and the dump $\dump$ and applying the unfolding. Note that, to improve readability, $\stack$ is decoded in postfix notation for plugging.
\end{varenumerate}\smallskip

\textbf{The Transitions.} According to the distillation approach of \cite{DBLP:conf/icfp/AccattoliBM14} (related to Danvy and Nielsen's \emph{refocusing} \cite{Danvy04refocusingin}) we distinguish different kinds of transitions, whose names reflect a proof-theoretical view, as transitions can be seen as cut-elimination steps \cite{DBLP:journals/toplas/AriolaBS09,DBLP:conf/icfp/AccattoliBM14}:
\begin{varenumerate}
	\item \emph{Multiplicatives} $\tomachm$: they correspond to the firing of a $\beta$-redex $(\la\var\tm)\tmtwo$, except that if the argument $\tmtwo$ is not a variable and the redex is not erasing (that happens when $\var \in \fv\tm$) then $\codetwo$ is not substituted but added to the environment;
	\item \emph{Exponentials} $\tomache$: they perform a clashing-avoiding substitution from the environment on the single variable occurrence represented by the current code. They implement micro-step substitution.
	\item \emph{Commutatives} $\tomachc$: they locate and
expose the next redex according to the maximal strategy, by rearranging the data-structures.
\end{varenumerate}

Both exponential and commutative transitions are invisible on the $\l$-calculus. Garbage collection of environment entries that are no longer necessary is here simply ignored, or, more precisely, it is encapsulated at the meta-level, in the decoding function. \smallskip

\textbf{Labels for Useful Sharing} A label $\lab$ for a code in the environment can be of three kinds. Roughly, they are:
\begin{varenumerate}
	\item \emph{Neutral}, or $\lab = \labneu$: it marks a neutral term, that is always useless as it is $\beta$-normal and its substitution cannot create a redex, because it is not an abstraction;
	\item \emph{Abstraction}, or $\lab = \lababs$: it marks an abstraction, that is a term that is at times useful to substitute. If the variable that it is meant to replace is applied, indeed, the substitution of the abstraction creates a $\beta$-redex. But if it is not applied, it is useless. 
	\item \emph{Redex}, or $\lab = \labredsym$: it marks a term that contains a $\beta$-redex. It is always useful to substitute these terms.
\end{varenumerate}
Actually, the explanation we just gave is oversimplified, but it provides a first intuition about labels. In fact in an environment $\esublab\var\code \cons \genv$ it is not really $\code$ that has the property mentioned by its label, rather the term $\relunf\code\genv$ obtained by unfolding the rest of the environment on $\code$. The idea is that $\esublabe\var\code\labredsym$ states that it is useful to substitute $\code$ to \emph{later on} obtain a redex inside it (by potential further substitutions on its variables coming from $\genv$). The precise meaning of the labels will be given by \refdef{well-lab-env}, and the properties they encode will be made explicit by \reflemma{exp-decr} in the Appendix (page \pageref{l:exp-decr}).

A further subtlety is that the label $\labredsym$ for redexes is refined as a pair $\labred n$, where $n$ is the number of substitutions in $\genv$ that are needed to obtain the maximal redex in $\relunf\code\genv$. Our machines never inspect these numbers, they are only used for the complexity analysis of \refsect{compl-anal}.\smallskip

\textbf{Grafting and Unfoldings} The unfolding of the environment $\genv$ on a code $\code$ is defined as the recursive  \emph{capture-allowing} substitution (called \emph{grafting}) of the entries of $\genv$ on $\code$. For lack of space, the precise definition has been moved to the Appendix (page \pageref{ssect:grafting}).

%% file: 03_-_The_Checking_Abstract_Machine.tex
\section{The Checking Abstract Machine}
\input{03b_-_checks-figure}

The Checking Abstract Machine (\ucam) is defined in \reffig{UCAM} and it is a variation over the very similar auxiliary checking machine for the Useful MAM in \cite{DBLP:conf/wollic/Accattoli16}. The difference between the two is in the two new transitions $\toucmachhole{\csym_7}$ and $\toucmachhole{\osym_6}$ (that is also why in \reffig{UCAM} they are misplaced with respect to the progressing numbering), plus the side condition \emph{if $\var\in\fv\code$} in $\toucmachhole{\osym_1}$. 

The \ucam starts executions on states of the form $\mac{\skeval}{\stempty}{\code}{\stempty}{\genv}$, with the aim of checking the usefulness of $\code$ with respect to the environment $\genv$, \ie it walks through $\code$ and whenever it encounters a variable $\var$ it looks up its usefulness in $\genv$. 

The \ucam has seven commutative transitions, noted $\toucmachhole{\csym_i}$ with $i=1,..,7$, used to walk through the term, and six output transitions, noted $\toucmachhole{\osym_j}$ with $j=1,..,6$, that produce the value of the test for usefulness, to be later used by the \maxmam. The exploration is done in two alternating phases, evaluation $\skeval$ and backtracking $\skback$. Evaluation explores the current code (morally towards the head, except when it encounters an erasing redex, that is when it first explores the right subterm) storing in the stack and in the frame the parts of the code that it leaves behind. Backtracking comes back to an argument that was stored in the frame, when the current head has already been checked. Note that the \ucam never modifies the environment, it only looks it up, nor it erases or duplicates any piece of code, it only walks through the data structures.

Let us explain the transitions. First the commutative ones:

\begin{itemize}
 \item \emph{$\toucmachasubone$}: the code is an application $\code\codetwo$ and the machine starts exploring the left subterm $\code$, storing $\codetwo$ on top of the stack $\stack$.
   \item \emph{$\toucmachasubseven$}: the code $\la\var\code$ and the first argument $\codetwo$ on the stack form an erasing redex (by hypothesis $\var \notin \fv\code$). The machine has found a $\beta$-redex but it cannot output yet, because it does not know what number $n$ to associate to the label $\lab = \labred n$. Indeed, $(\la\var\code)\codetwo$ is the next $\beta$-redex to reduce only if $\codetwo$ is normal, otherwise the next one will be in $\codetwo$ and obtaining it may require many substitution steps. Therefore, the machine stores $\la\var\code$ and the current stack $\stack$ on the frame $\skframe$ and starts checking $\codetwo$ (with an empty stack).
 \item \emph{$\toucmachasubtwo$}: the code is an abstraction $\la\var\code$ and the machine goes under the abstraction, storing $\var$ on top of the frame $\mframe$.  
  \item \emph{$\toucmachasubthree$}: the machine finds a variable $\var$ that either has no associated entry in the environment (if $\genv(\var) = \bot$) or its  associated entry $\esublab\var\code$ in the environment is useless. This can happen if either $\lab = \labneu$, \ie substituting $\code$ would only lead to a neutral term, or $\lab = \lababs$, \ie substituting $\code$ would provide an abstraction, but the stack is empty, and so it is useless to substitute the abstraction because no $\beta$-redexes will be obtained. Thus the machine switches to the backtracking phase ($\skback$), whose aim is to undo the frame to obtain a new subterm to explore.
  \item \emph{$\toucmachasubfour$}: it is the inverse of $\toucmachasubtwo$, it puts back on the code an abstraction that was previously stored in the frame.
  \item \emph{$\toucmachasubfive$}: backtracking from the evaluation of an argument $\codetwo$, it restores the application $\code\codetwo$ and the stack $\stack$ that were previously stored in the frame.
  \item \emph{$\toucmachasubsix$}: backtracking from the evaluation of the left subterm $\code$ of an application $\code\codetwo$, the machine starts evaluating the right subterm (by switching to the evaluation phase $\skeval$) with an empty stack $\stempty$, storing on the frame the pair $\dentry\code\stack$ of the left subterm and the previous stack $\stack$.

\end{itemize}

Then the output transitions:
\begin{itemize}
 \item \emph{$\toucmachoutone$}: the machine finds a non-erasing $\beta$-redex, namely $(\la\var\code)\codetwo$ (by hypothesis $\var \in \fv\code$, and thus outputs a label saying that it requires only one substitution step (namely substituting the term the machine was executed on) to eventually find a $\beta$-redex.
 \item \emph{$\toucmachouttwo$}: the machine finds a variable $\var$ whose associated entry $\esublabe\var\code{\labred n}$ in the environment is labeled with $\labred n$, and so outputs a label saying that it takes $n+1$ substitution steps to eventually find a $\beta$-redex ($n$ plus 1 for the term the machine was executed on).  
  \item \emph{$\toucmachoutthree$}: the machine finds a variable $\var$ whose associated entry $\esublabe\var\code\lababs$ in the environment is labeled with $\lababs$, so $\code$ is an abstraction, and the stack is non-empty. Since substituting the abstraction will create a $\beta$-redex, the machine outputs a label saying that it takes two substitution steps to obtain a $\beta$-redex, one for the term the machine was executed on and one for the abstraction $\code$.
    \item \emph{$\toucmachoutsix$}: the machine went through the whole code $\codetwo$ and found no redexes. Then $\codetwo$ is normal and together with the erasing abstraction $\la\var\code$ in the frame it forms an erasing redex ready to fire, and so the output is a label saying that it requires only one substitution step (namely substituting the term the machine was executed on) to eventually find a $\beta$-redex.
  \item \emph{$\toucmachoutfour$}: the machine went through the whole term, that is an application, and found no redex, nor any redex that can be obtained by substituting from the environment. Thus that term is neutral and so the machine outputs the corresponding label.
  \item \emph{$\toucmachoutfive$}: as for the previous transition, except that the term is an abstraction, and so the output is the $\lababs$ label.

\end{itemize}

The fact that commutative transitions only walk through the code, without changing anything, is formalized by the following lemma, that is crucial for the proof of correctness of the \ucam (forthcoming \refthm{ucam-propert}).

\begin{lemma}[Commutative Transparency\withproofs{, Proof at P. \pageref{ssect:comm-transp}}]
  \label{l:comm-transp} 
	Let $\state = \mac{\skphase}{\skframe}{\code}{\stack}{\genv} \tomachhole{\csym_{1,2,3,4,5,6,7}}\mac{\skphase'}{\skframe'}{\code'}{\stack'}{\genv} =\statetwo$. Then
	\begin{enumerate}
	 \item \emph{Decoding Without Unfolding}: \label{p:comm-transp-trunk} $\decodep\mframe{\decodestack{\code}{\stack}} = \decodep{\mframe'}{\decodestack{\code'}{\stack'}}$, and
	 \item \emph{Decoding With Unfolding}: \label{p:comm-transp-state}
	 $\decode\state = \decode\statetwo$.
	\end{enumerate}	
 \end{lemma}

For the analysis of the properties of the \ucam we need a notion of well-labeled environment, \ie of environment where the labels are consistent with their intended meaning. It is a technical notion also providing enough information to perform the complexity analysis, later on. Moreover, it includes two structural properties of environments: 1) in $\esublab\var\code$ the code $\code$ cannot be a variable, and 2) there cannot be two entries associated to the same variables.

\begin{definition}[Well-Labeled Environments]
\label{def:well-lab-env}
  \emph{Well-labeled environments} $\genv$ are defined by		 
  \begin{enumerate}
    \item \emph{Empty}: $\stempty$ is well-labeled;
    \item \emph{Inductive}: $\esublab\var\code \cons\genvtwo$ is well-labeled if $\genvtwo$ is well-labeled, $\var$ is fresh with respect to $\code$ and $\genvtwo$, and 
    \begin{enumerate}
     \item \emph{Abstractions}: 
      if $\lab = \lababs$ then $\code$ and $\relunf\code\genvtwo$ are normal abstractions; 
						
    \item \emph{Neutral Terms}:
      if $\lab = \labneu$ then $\code$ is an application and $\relunf\code{\genvtwo}$ is neutral.
      
          \item \emph{Redexes}: if $\lab = \labred n$ then $\code$ is not a variable, $\relunf\code{\genvtwo}$ contains a $\beta$-redex. Moreover, there is a \perp-context $\ctx$ such that $\code = \ctxp\codetwo$ and 
          \begin{itemize}
          \item if $n= 1$ then $\codetwo$ is a $\rtobperp$-redex,
          \item if $n >1 $ then $\codetwo = \var$ and $\genvtwo = \genvthree\cons\esublab\vartwo\codetwo\cons\genvfour$ with 
          \begin{itemize}
          \item if $n> 2$ then $\lab = \labred {n-1}$
          \item if $n = 2$ then  $\lab = \labred {1}$ or ($\lab = \lababs$ and $\ctx$ is applicative).
          \end{itemize}
          \end{itemize}
    \end{enumerate}    
   \end{enumerate}
\end{definition}


The study of the \ucam requires some technical invariants proved in the appendix (page \pageref{ssect:ucam-invariants}). The next theorem provides the main properties of the \ucam, \ie that when executed on $\code$ and $\genv$ it provides a label $\lab$ to extend $\genv$ with a consistent entry for $\code$ (\ie such that $\esublab\var\code\cons\genv$ is well-labeled), and that such an execution takes time linear in the size of $\code$. 

Let us explain an important point about the complexity of the machine transition. Variables are meant to be implemented as memory locations and variable occurrences as pointers to those location. Therefore, the global environment $\genv$ is a store and can be accessed \emph{randomly}, that is, with no need to go through it sequentially. With this hypothesis on the representation of terms, see \cite{DBLP:conf/ppdp/AccattoliB17} for an actual implementation, all the transitions of the \ucam but $\toucmachasubseven$ and $\toucmachoutone$ can be implemented in constant time. The check $\var\in/\notin \fv\code$ in $\toucmachasubseven$ and $\toucmachoutone$ at first sight requires time proportional to the size of the initial term, but in fact it can be implemented in $O(1)$ if one assumes a stronger hypothesis on the representation of terms: a variable is a data-type with a memory location \emph{plus} pointers to its occurrences. Then all transitions can be implemented in $O(1)$. Anyway, this hypothesis is not essential for the implementation to be reasonable. If dropped, the \emph{subterm property} (\reflemmap{umam-quant-invs}{subterm}) guarantees that the size of terms does not explode and that the check $\var\in/\notin \fv\code$ can be done in reasonable time. The hypothesis, however, simplifies considerably the analysis.

\begin{theorem}[\ucam Properties\withproofs{, Proof at Page \pageref{ssect:ucam-propert}}]
\label{thm:ucam-propert} 
Let $\code$ be a code and $\genv$ a global environment. 
\begin{enumerate}
	\item \emph{Determinism and Progress}: the \ucam is deterministic and there always is a transition that applies;
	\item \label{p:ucam-propert-compl}
	\emph{Termination and Complexity}: the execution of the \ucam on $\code$ and $\genv$ always terminates, taking $O(\size{\code})$ steps, moreover 
	\item \label{p:ucam-propert-correc}
	\emph{Correctness}: if $\genv$ is well-labeled, $\var$ is fresh with respect to $\genv$ and $\code$, and $\lab$ is the output then $\esublab\var\code\cons\genv$ is well-labeled. 
\end{enumerate}
\end{theorem}

%% file: 03b_-_checks-figure.tex
\begin{figure}[t]
	\begin{center}
	\ovalbox{$\begin{array}{l@{\hspace{1em}}rcl@{\hspace{2em}}l@{\hspace{1em}}rcl}
			\text{Frames} &
			\skframe & \grameq & \stempty
								\mid \skframe\cons \stackitem
			&
			\text{Stacks} &
			\stack   & \grameq & \stempty \mid \code\cons\stack
			\\
			\text{Frame Items} & \stackitem & \grameq & \var \mid \dentry{\code}{\stack} \mid \gcentry{ \la\var\code }\stack
			&
			\text{Phases} &
			\skphase & \grameq & \skeval \mid \skback
			\\
			\text{Labels} &
			\lab    & \grameq & \lababs \mid (\labredsym, n\in \nat) \mid \labneu
&
			\text{Environments} &
			\genv    & \grameq & \stempty
								\mid \esublab{\var}{\code}\cons\genv
		\end{array}$}
	\end{center}
	\caption{Grammars.}
	\label{fig:grammars}
\end{figure}
	\begin{figure}[t]
	\begin{center}
	  \ovalbox{
		\scalemath{0.94}{
		\setlength{\arraycolsep}{0.5em}
		\begin{array}{c|c|c|c|clc|c|c|c|c}
		\mbox{Frame} & \mbox{Code} & \mbox{Stack} & \mbox{Env} & \mbox{Ph}
		&&
		\mbox{Frame} & \mbox{Code} & \mbox{Stack} & \mbox{Env} & \mbox{Ph}\\
			\maca{ \skeval }{ \skframe }{ \code\codetwo }{ \stack }{ \genv } &
			\toucmachasubone &
			\maca{ \skeval }{ \skframe }{ \code }{ \codetwo\cons\stack }{ \genv }
		\\
			\maca{ \skeval }{ \skframe }{ \la\var\code }{ \codetwo\cons\stack }{ \genv } &
			\toucmachoutone &
			\multicolumn{5}{l}{\text{output $\labred 1$}}\\
			\multicolumn{11}{r}{\text{ if $\var \in \fv\code$}}
			\\
			\maca{ \skeval }{ \skframe }{ \la\var\code }{ \codetwo\cons\stack }{ \genv } &
			\tomachasubseven &
			\maca{ \skeval }{ \skframe \cons \gcentry{ \la\var\code }\pi}{ \codetwo }{ \stempty }{ \genv }		\\
			\multicolumn{11}{r}{\text{ if $\var \notin \fv\code$}}
		\\
			\maca{ \skeval }{ \skframe }{ \la\var\code }{ \stempty }{ \genv } &
			\toucmachasubtwo &
			\maca{ \skeval }{ \skframe\cons\var }{ \code }{ \stempty }{ \genv }
		\\
			\maca{ \skeval }{ \skframe }{ \var }{ \stack }{ \genv } &
			\toucmachouttwo &
			\multicolumn{5}{l}{\text{output $\labred {n+1}$}} \\
			\multicolumn{11}{r}{\text{ if $\genv(\var) = \esublabe\var\code{\labred n}$}}
		\\
			\maca{ \skeval }{ \skframe }{ \var }{ \codetwo\cons\stack }{ \genv } &
			\toucmachoutthree &
			\multicolumn{5}{l}{\text{output $\labred 2$}} \\
			\multicolumn{11}{r}{\text{ if $\genv(\var) = \esublabe\var\code\lababs$}}

		\\
			\maca{ \skeval }{ \skframe }{ \var }{ \stack }{ \genv } &
			\toucmachasubthree &
			\maca{ \skback }{ \skframe }{ \var }{ \stack }{ \genv }
			\\ \multicolumn{11}{r}{\text{ if $\genv(\var) = \bot$ or $\genv(\var) = \esublabe\var\code\labneu$ or ($\genv(\var) = \esublabe\var\code\lababs$ and $\stack = \stempty$)}}
		\\
			\maca{ \skback }{ \skframe\cons\var }{ \code }{ \stempty }{ \genv } &
			\toucmachasubfour &
			\maca{ \skback }{ \skframe }{ \l\var.\code }{ \stempty }{ \genv }
		\\

			\maca{ \skback }{ \skframe\cons \dentry\code\stack  }{ \codetwo }{ \stempty }{ \genv } &
			\toucmachasubfive &
			\maca{ \skback }{ \skframe }{ \code\codetwo }{ \stack }{ \genv }
		\\
			\maca{ \skback }{ \skframe\cons \gcentry{ \la\var\code }\stack  }{ \codetwo }{ \stempty }{ \genv } &
			\toucmachoutsix &
			\multicolumn{5}{l}{\text{output $\labred 1$}}
		\\
			\maca{ \skback }{ \skframe }{ \code }{ \codetwo\cons\stack }{ \genv } &
			\toucmachasubsix &
			\maca{ \skeval }{ \skframe \cons \dentry\code\stack }{ \codetwo }{ \stempty }{ \genv }
		\\
			\maca{ \skback }{ \stempty }{ \code\codetwo }{ \stempty }{ \genv } &
			\toucmachoutfour &
			\multicolumn{5}{l}{\text{output $\labneu$}} 
		\\
			\maca{ \skback }{ \stempty }{ \la\var\code }{ \stempty }{ \genv } &
			\toucmachoutfive &
			\multicolumn{5}{l}{\text{output $\lababs$}} 

		\end{array}
		}
		}
	\end{center}
	
	\caption{The Checking Abstract Machine (\ucam).}
	\label{fig:UCAM}
\end{figure}

\begin{figure}[t]
\begin{center}
     \ovalbox{
$\begin{array}{rcllrcllllllll}
	\decode{\stempty}	& \defeq & \ctxhole &&	
	\decode{\mframe\cons\var} & \defeq & \decdumpp{\la\var\ctxhole}  \\
        
	\decode{ \codetwo \cons \stack} 			& \defeq & \decstackp{\ctxhole\codetwo} &&
	\decode{\stctx{\state}}				& \defeq & \relunf{\decdump\ctxholep{\decstack}} \genv  \\
	
         \decode{\mframe\cons\dentry\code\stack} & \defeq & \decdumpp{\decstackp{\code\ctxhole}} &&
         \decode{\state}				& \defeq & \relunf{ \decdump\ctxholep{\decstackp\code} }\genv & = &
	\decodep{\stctx{\state}}{\relunf\code\genv} &  
	\\
	\decode{\mframe\cons \gcentry{ \la\var\code }\stack} & \defeq & \decdumpp{\decstackp{(\la\var\code)\ctxhole}} &&&\multicolumn{5}{l}{\mbox{where $\state = \ucamstate\mframe\code\stack\genv$}}
	\\
              
\end{array}$
}
  \end{center}  
	\caption{Decoding.}
	\label{fig:decoding}
\end{figure}

%% file: 04_-_The_Maximal_Milner_Abstract_Machine.tex
\section{The Maximal Milner Abstract Machine}
\label{sect:SMAM}
\input{04b_-_maximal-machine-figure}

The Maximal Milner Abstract Machine, or \maxmam, is defined in \reffig{UMAM} and it is a small variation over the Useful MAM of \cite{DBLP:conf/wollic/Accattoli16} (as for the \ucam, the difference is in the two new transitions, $\tomachasubseven$ and $\tomachmthree$, plus the side condition \emph{if $\var\in\fv\code$} in $\tomachmtwo$). It is very similar to the \ucam, in particular it has exactly the same commutative transitions, and the same organization in evaluating and backtracking phases. The difference with respect to the \ucam is that the output transitions are replaced by micro-step computational rules that reduce $\beta$-redexes and implement useful substitutions. Let us explain them:
\begin{varitemize}
 \item \emph{Multiplicative Transition $\tomachmone$}: when the argument of the $\beta$-redex $(\la\var\code)\vartwo$ is a variable $\vartwo$ then it is immediately substituted in $\code$. This happens because 1) such substitution are not costly (by the subterm invariant, \reflemma{umam-quant-invs} in the Appendix, their cost is bound by the size of $\code$ that is bound by the size of the initial term); 2) because in this way the environment stays compact; 3) because in this way the labels for useful sharing are slightly simpler.

 \item \emph{Multiplicative Transition $\tomachmtwo$}: since the argument $\codetwo$ is not a variable and the redex is not erasing ($\var \in \fv\code$ by hypothesis) then the redex is fired by adding the entry $\esublab\var\codetwo$ to the environment, with $\lab$ obtained by running the \ucam on $\codetwo$ and $\genv$.

 \item \emph{Multiplicative Transition $\tomachmthree$}: an invariant of the machine is that when it backtracks (phase $\skback$) the code is normal (see \reflemmap{umam-invariants}{nf-code} in the Appendix, page \pageref{l:umam-invariants}). Moreover, every abstraction $\la\var\code$ in an erasing redex entry $\gcentry{\la\var\code}\stack$ in the frame is such that $\var \notin \fv\code$ (\reflemmap{umam-invariants}{erasing}). Then in $\tomachmthree$ the code $\codetwo$ is normal and together with $\la\var\code$ it forms an erasing $\rtobperp$ redex. Correctly, the machine throws away $\codetwo$ and switches to evaluating $\code$ (note the change of phase from $\skback$ to $\skeval$).

 \item \emph{Exponential Transition $\tomachered$}: the environment entry associated to $\var$ is labeled with $\labred n$ thus it is useful to substitute $\code$. The idea is that in at most $n$ additional substitution steps (shuffled with commutative steps) a $\beta$-redex will be obtained. To avoid variable clashes the substitution $\alpha$-renames $\code$.
 \item \emph{Exponential Transition $\tomacheabs$}: the environment associates an abstraction to $\var$ and the stack is non empty, so it is useful to substitute the abstraction (again, $\alpha$-renaming to avoid variable clashes). Note that if the stack is empty the machine rather backtracks using $\tomachasubthree$.
\end{varitemize}

The \maxmam starts executions on \emph{initial states} of the form $\mac{\skeval}{\stempty}{\code}{\stempty}{\stempty}$, where $\code$ is such that any two variables (bound or free) have distinct names, and any other component is empty. A state $\state$ is \emph{reachable} if there are an initial state $\statetwo$ and a \maxmam execution $\exec:\statetwo \tomach^*\state$, and it is \emph{final} if no transitions apply.

The theorem of correctness and completeness of the machine with respect to $\toperp$ follows. It rests on a technical development that is in the appendix, starting at page \pageref{ssect:umam-invariants}. The bisimulation is \emph{weak} because transitions other than $\tomachm$ are invisible on the $\l$-calculus. For a machine execution $\exec$ we denote with $\size\exec$ (resp. $\sizex\exec$) the number of transitions (resp. $\mathtt{x}$-transitions for $\mathtt{x} \in \set{\msym, \esym, \csym, \ldots}$) in $\exec$.

\begin{theorem}[Weak Bisimulation, Proof at Page \pageref{ssect:weak-bis}]
\label{thm:weak-bis} 
  Let $\state$ be an initial \maxmam state of code $\code$.
  
 \begin{enumerate}
  \item \emph{Simulation:} for every execution $\exec:\state\tomach^*\statetwo$ there exists a derivation $\deriv \colon \decode\state \toperp^*\decode\statetwo$ such that $\size\deriv = \sizem\exec$;
  
  \item \label{p:weak-bis-rev-sim} \emph{Reverse Simulation:} for every derivation $\deriv \colon \code\toperp^*\tmtwo$ there is an execution $\exec:\state\tomach^*\statetwo$ such that $\decode\statetwo = \tmtwo$ and $\size{\deriv} = \sizem\exec$.
 \end{enumerate}
\end{theorem}

%% file: 04b_-_maximal-machine-figure.tex
\begin{figure}[t]
	\begin{center}
	  \ovalbox{
		\scalemath{0.94}{
		\setlength{\arraycolsep}{0.5em}
		\begin{array}{c|c|c|c|clc|c|c|c|c}
		\mbox{Frame} & \mbox{Code} & \mbox{Stack} & \mbox{Env} & \mbox{Ph}
		&&
		\mbox{Frame} & \mbox{Code} & \mbox{Stack} & \mbox{Env} & \mbox{Ph}\\
			\maca{ \skeval }{ \skframe }{ \code\codetwo }{ \stack }{ \genv } &
			\tomachasubone &
			\maca{ \skeval }{ \skframe }{ \code }{ \codetwo\cons\stack }{ \genv }
		\\
			\maca{ \skeval }{ \skframe }{ \la\var\code }{ \vartwo\cons\stack }{ \genv } &
			\tomachmone &
			\maca{ \skeval }{ \skframe }{ \code\isub\var\vartwo }{ \stack }{ \genv }
		\\
			\maca{ \skeval }{ \skframe }{ \la\var\code }{ \codetwo\cons\stack }{ \genv } &
			\tomachmtwo &
			\maca{ \skeval }{ \skframe }{ \code }{ \stack }{ \esublab{\var}{\codetwo}\cons\genv }
		\\
			\multicolumn{11}{r}{\text{ if $\var \in \fv\code$ and $\codetwo$ is not a variable and $\lab$ is the output of the \ucam on $\codetwo$ and $\genv$}}
		\\
			\maca{ \skeval }{ \skframe }{ \la\var\code }{ \codetwo\cons\stack }{ \genv } &
			\tomachasubseven &
			\maca{ \skeval }{ \skframe \cons \gcentry{ \la\var\code }\pi}{ \codetwo }{ \stempty }{ \genv }
		\\
			\multicolumn{11}{r}{\text{ if $\var \notin \fv\code$ and $\codetwo$ is not a variable}}
		\\
			\maca{ \skeval }{ \skframe }{ \l\var.\code }{ \stempty }{ \genv } &
			\tomachasubtwo &
			\maca{ \skeval }{ \skframe\cons\var }{ \code }{ \stempty }{ \genv }
		\\
			\maca{ \skeval }{ \skframe }{ \var }{ \stack }{ \genv } &
			\tomachered &
			\maca{ \skeval }{ \skframe }{ \rename{\code} }{ \stack }{ \genv } \\
			\multicolumn{11}{r}{\text{ if $\genv(\var) = \esublabe\var\code{\labred n}$}}
		\\
			\maca{ \skeval }{ \skframe }{ \var }{ \codetwo\cons\stack }{ \genv } &
			\tomacheabs &
			\maca{ \skeval }{ \skframe }{ \rename{\code} }{ \codetwo\cons\stack }{ \genv } \\
			\multicolumn{11}{r}{\text{ if $\genv(\var) = \esublabe\var\code\lababs$}}

		\\
			\maca{ \skeval }{ \skframe }{ \var }{ \stack }{ \genv } &
			\tomachasubthree &
			\maca{ \skback }{ \skframe }{ \var }{ \stack }{ \genv }
			\\ \multicolumn{11}{r}{\text{ if $\genv(\var) = \bot$ or $\genv(\var) = \esublabe\var\code\labneu$ or ($\genv(\var) = \esublabe\var\code\lababs$ and $\stack = \stempty$)}}
		\\
			\maca{ \skback }{ \skframe\cons\var }{ \code }{ \stempty }{ \genv } &
			\tomachasubfour &
			\maca{ \skback }{ \skframe }{ \l\var.\code }{ \stempty }{ \genv }
		\\

			\maca{ \skback }{ \skframe\cons \dentry\code\stack  }{ \codetwo }{ \stempty }{ \genv } &
			\tomachasubfive &
			\maca{ \skback }{ \skframe }{ \code\codetwo }{ \stack }{ \genv }
		\\
			\maca{ \skback }{ \skframe\cons \gcentry{ \la\var\code }\stack  }{ \codetwo }{ \stempty }{ \genv } &
			\tomachmthree &
			\maca{ \skeval }{ \skframe }{ \code }{ \stack }{ \genv }
		\\
			\maca{ \skback }{ \skframe }{ \code }{ \codetwo\cons\stack }{ \genv } &
			\tomachasubsix &
			\maca{ \skeval }{ \skframe \cons \dentry\code\stack }{ \codetwo }{ \stempty }{ \genv }
				
		\end{array}
		}}
		
		\scriptsize$\rename{\code}$ is any code $\alpha$-equivalent to $\code$ such that it is well-named and its bound names are fresh with respect to those in the other machine components.
	\end{center}
	\caption{The Maximum (Useful) Milner Abstract Machine (\maxmam).}
	\label{fig:UMAM}
\end{figure}

%% file: 05_-_Quantitative_Analysis.tex
\section{Quantitative Analysis}
\label{sect:compl-anal}

The complexity analysis of the \maxmam is omitted because it follows exactly, by changing only minimal details, the one for the Useful MAM in \cite{DBLP:conf/wollic/Accattoli16}. All the details can be found in the Appendix, starting from page \pageref{ssect:quant-analysis}.

Let us anyway provide the schema of the analysis. First of all one proves a subterm invariant, proving that most codes in the state of the \maxmam are subcodes of the initial code. Then the proof of the polynomial bound of the overhead is in three steps. 
\begin{varenumerate}
\item \emph{Exponential vs Multiplicative Transitions}: we bound the number $\sizee\exec$ of exponential transitions of an execution $\exec$ using the number $\sizem\exec$ of multiplicative transitions of $\exec$, that by \refthm{weak-bis} corresponds to the number of maximal $\beta$-steps on the $\l$-calculus. The bound is quadratic.

\item \emph{Commutative vs Exponential Transitions}: we bound the number $\sizecom\exec$ of commutative transitions of $\exec$ by using the number of exponential transitions and the size of the initial term. The bound is linear in both quantities.

\item \emph{Global bound}: we multiply the number of each kind of transition for the cost of that kind (everything is constant time but for exponential  transitions, that are linear in the size of the initial term), and then sum over the kind of transitions.
\end{varenumerate}

Concretely, one obtains the following theorem.

\begin{theorem}[\maxmam Overhead Bound, Proof at Page \pageref{thm:final-thm}]
\label{thm:eglamour-overhead-bound}
  Let $\deriv:\tm \toperp^* \tmtwo$ be a maximal derivation and $\exec$ be the \maxmam execution simulating $\deriv$ given by \refthmp{weak-bis}{rev-sim}. Then:
  \begin{enumerate}
   \item \emph{Length}: $\size\exec = O((1+\size{\deriv}^2)\cdot{\size\tm})$.
   
   \item \emph{Cost}: $\exec$ is implementable on RAM in $O((1+\size{\deriv}^2)\cdot{\size\tm})$ steps.
  \end{enumerate}
\end{theorem}

%% file: 99_-_appendix-minor.tex
\section*{Proof Appendix}

\subsection{Proof of the properties of the Maximal Strategy (\reflemmaeq{perp-properties}, p. \pageref{l:perp-properties})}
 \label{ssect:perp-properties}

\begin{proof}
 By induction on $\tm$. Cases:
 \begin{itemize}
   \item \emph{Variable}, \ie $\tm = \var$. Then $\tm$ is normal, absurd.
   
   \item \emph{Abstraction}, \ie $\tm = \la\var\tmthree$. Then by \ih there exists $\tmfour$ such that $\tmthree \toperp \tmfour$. By rule \perpRuleLambda, $\tm = \la\var\tmthree \toperp \la\var\tmfour$, \ie just take $\tmtwo \defeq \la\var\tmfour$. \emph{Determinism}: it follows from the \ih
   
   \item \emph{Application}, \ie $\tm = \tmthree \tmfour$. Two cases:
   \begin{itemize}
    \item \emph{$\tmthree$ is an abstraction}, \ie $\tmthree = \la\var\tmfive$. Two sub-cases:
    \begin{itemize}
      \item \emph{$\var\in\tmthree$ or $\tmfour$ is normal}. Then $\tm = (\la\var\tmfive) \tmfour \rtobperp \tmfive \isub\var\tmfour$. \emph{Determinism}: the rules for \perp contexts do not allow to evaluate in $\tmthree$, because $\la\varthree \tmthree$ is applied, nor to evaluate in $\tmfour$ (if it is not normal) because then $\var \in \tmthree$ and so rule \perpRuleGc\ cannot be applied.
      
      \item \emph{$\var\notin\tmthree$ and $\tmfour$ is not normal}. By \ih there exists a unique $\tmsix$ such that $\tmfour \toperp \tmsix$, that is, there is a perpetual context $\ctx$ such that $\tmfour = \ctxp{\tmfour'} \toperp \ctxp{\tmsix'}$ with $\tmfour' \rtobperp \tmsix'$. By rule \perpRuleGc, $(\la\var\tmfive) \ctx$ is a \perp-context and $\tm = (\la\var\tmfive) \ctxp{\tmfour'} \toperp (\la\var\tmfive) \ctxp{\tmsix'}$. Clearly, there cannot be any $\toperp$ redex in $\tmthree$, so determinism holds.
    \end{itemize}
    
    \item \emph{$\tmthree$ is not an abstraction}. Two sub-cases:
    \begin{itemize}
      \item \emph{$\tmthree$ is not normal}. Then by \ih there exists unique $\tmfive$ such that $\tmthree \toperp \tmfive$, By rule \perpRuleAppL, $\tm = \tmthree \tmfour \toperp \tmfive\tmfour$. Since $\tmthree$ reduces, it is not neutral and so there cannot be $\toperp$ redexes in $\tmfour$ (because rule \perpRuleGc does not apply).

      \item \emph{$\tmthree$ is normal and thus neutral}. Then $\tmfour$ is not normal (otherwise $\tm$ is normal, absurd). By \ih there exists unique $\tmfive$ such that $\tmfour \toperp \tmfive$. By rule \perpRuleAppR, $\tm = \tmthree \tmfour \toperp \tmthree \tmfive$, and since $\tmthree$ is neutral this is the unique $\toperp$ redex of $\tm$.
    \end{itemize}
   \end{itemize}
   
 \end{itemize}

\end{proof}

\subsection{Definition of Grafting and Unfolding, and their Properties}
\label{ssect:grafting}

The unfolding of the environment $\genv$ on a code $\code$ is defined as the recursive  \emph{capture-allowing} substitution (called \emph{grafting}) of the entries of $\genv$ on $\code$.

\begin{definition}[Grafting and Environment Unfolding]
 The operation of grafting $\code\csub\var\codetwo$ is defined by 
 \[\begin{array}{rllllllrlll}   
 (\codethree\codefour)\csub\var\codetwo & \defeq & \codethree\csub\var\codetwo \codefour\csub\var\codetwo  &&&&&
    (\la\vartwo\codethree)\csub\var\codetwo & \defeq &  \la\vartwo\codethree\csub\var\codetwo \\
    \var\csub\var\codetwo & \defeq & \codetwo &&&&&
    \vartwo\csub\var\codetwo  & \defeq & \vartwo
    \end{array}\]
 Given an environment $\genv$ we define the unfolding of $\genv$ on a code $\code$ as follows:
 \[\begin{array}{rllllllrlll}    
    \relunf{\code}{\stempty} & \defeq & \code &&&&&
    \relunf{\code}{\esublab\var\codetwo \cons \genv} & \defeq & \relunf{\code\csub\var\codetwo}\genv    
    \end{array}\]
 or equivalently as:
 \[\begin{array}{rllllllrlllllllllll}    
    \relunf{(\codetwo\codethree)}\genv & \defeq & \relunf\codetwo\genv \relunf\codethree\genv&&&&&
    \relunf\var{\esublab\var\codetwo\cons\genvtwo} & \defeq & \relunf{\codetwo}{\genvtwo} \\
    \relunf{(\la\var\codetwo)}\genv & \defeq & \la\var\relunf\codetwo\genv &&&&&
    \relunf\var{\esublab\vartwo\codetwo\cons\genvtwo} & \defeq & \relunf\var\genvtwo &&&&&\relunf\var{\stempty} & \defeq & \var
    \end{array}\]
\end{definition}
For instance,  $\relunf{(\la\var\vartwo)}{\esublabe\vartwo{\var\var}\labneu} = \la\var(\var\var)$. The unfolding is extended to contexts as expected (\ie recursively propagating the unfolding and setting $\relunf\ctxhole\genv = \genv$).

Let us explain the need for grafting. In \cite{DBLP:conf/aplas/AccattoliBM15}, the Strong MAM is decoded to the LSC, that is a calculus with explicit substitutions, \ie a calculus able to represent the environment of the Strong MAM. Matching the representation of the environment on the Strong MAM and on the LSC does not need grafting but it is, however, a quite technical affair. Useful sharing adds many further complications in establishing such a matching, because useful evaluation computes a shared representation of the normal form and forces some of the explicit substitutions to stay under abstractions. The difficulty is such, in fact, that we found much easier to decode directly to the $\l$-calculus rather than to the LSC. Such an alternative solution, however, has to push the substitution induced by the environment through abstractions, which is why we use grafting.

The following easy properties will be used to prove the correctness of the machine (in \reflemma{one-step-weak-sim}).

\begin{lemma}[Properties of Grafting and Unfolding]
\label{l:grft-unf-proper} 
\hfill
\begin{enumerate}
  \item \label{p:grft-unf-proper-one} If the bound names of $\tm$ do not appear free in $\tmtwo$ then $\tm\isub\var\tmtwo = \tm\csub\var\tmtwo$.
  \item \label{p:grft-unf-proper-two} If moreover they do not appear free in $\genv$ then $\relunf\tm\genv \isub\var{\relunf\tmtwo\genv} = \relunf{\tm\isub\var\tmtwo} \genv$.
 \end{enumerate}
\end{lemma}

\subsection{Proof of the Commutative Transparency Lemma (\reflemmaeq{comm-transp}, p. \pageref{l:comm-transp})}
\label{ssect:comm-transp}

\input{\proofspath/commutative-distallation-proof}

\subsection{The \ucam Invariants Lemma (\reflemmaeq{ucam-invariants}, p. \pageref{l:ucam-invariants})}
\label{ssect:ucam-invariants}

Before the invariants we need a lemma to be used in the proof of the decoding invariant.

\begin{lemma}
\label{l:lo-properties} 
Let $\ctx$ be a context.
\begin{enumerate}
\item \emph{Right Extension}: \label{p:lo-properties-arg}
$\ctxp{\cdot}$ is \perp iff $\ctxp{\ctxhole\codetwo}$ is \perp;

\item \emph{Left Neutral Extension}: \label{p:lo-properties-neut}
$\ctxp{\cdot}$ is \perp and $\codetwo$ is neutral iff $\ctxp{\codetwo \ctxhole}$ is \perp;

\item \emph{Abstraction Extension}: \label{p:lo-properties-abs-app}
$\ctxp{\cdot}$ is \perp and not applicative iff $\ctxp{\la\var\ctxhole}$ is \perp.

\item \emph{Erasing Redexes Extension}: \label{p:lo-properties-erasing-red}
$\ctxp{\cdot}$ is \perp and $\var \notin \fv\tm$ iff $\ctxp{(\la\var\tm)\ctxhole}$ is \perp.

\item \emph{Unfolding Removal}: \label{p:lo-properties-unf-rem}
if $\relunf\ctx\genv$ is \perp then $\ctx$ is \perp.

\end{enumerate}
\end{lemma}

\begin{proof}
  By induction on the predicate \emph{$\ctx$ is \perp} (\refdef{ilob-ctx}, page \pageref{def:ilob-ctx}).
\end{proof}

Now, some terminology:
\begin{itemize}
\item A state $\state$ is \emph{initial} if it is of the form $\mac{\skeval}{\stempty}{\code}{\stempty}{\genv}$ with $\genv$ well-labeled, $\code$ well-named and such that the variables abstracted in $\code$ do not occur in $\genv$. 
\item A state $\state$ is \emph{reachable} if there are an initial state $\statetwo$ and a \ucam  execution $\exec:\statetwo \toucmach^*\state$. 
\end{itemize}

Finally:

\begin{lemma}[\ucam Invariants\withproofs{, Proof at P. \pageref{ssect:ucam-invariants}}]
\label{l:ucam-invariants} 
Let $\skamstate{\skphase}{\skframe}{\codetwo}{\stack}{\genv}$ be a \ucam state reachable from an initial state of code $\code_0$. Then 
\begin{enumerate}
 \item \emph{Normal Form}: \label{p:ucam-invariants-nf} 
		\begin{enumerate}
			\item \emph{Backtracking Code}: \label{p:ucam-invariants-nf-code} 
			if $\skphase = \skback$, then $\relunf\codetwo\genv$ is normal, and if $\stack$ is non-empty, then $\relunf\codetwo\genv$ is neutral; 
			
			\item \emph{Frame}: \label{p:ucam-invariants-nf-frame} 			
			  if $\skframe = \skframetwo\cons \skap\codethree\stacktwo \cons \skframethree$, then $\relunf\codethree\genv$ is neutral.	
		\end{enumerate}
  \item \emph{Subterm}:
 \label{p:ucam-invariants-subterm} 

  \begin{enumerate}
  \item \emph{Evaluating Code}: \label{p:ucam-invariants-subterm-one}
  if $\skphase = \skeval$, then $\codetwo$ is a subterm of $\code_0$;
  
  \item \emph{Stack}:  \label{p:ucam-invariants-subterm-two}
  any code  in the stack $\stack$ is a subterm of $\code_0$;
  
  \item \emph{Frame}: \label{p:ucam-invariants-subterm-three}
    \begin{enumerate}
      \item \emph{Head Contexts}: if $\skframe = \skframetwo\cons\skap{\codethree}{\stacktwo}\cons\skframethree$, then any code in $\stacktwo$ is a subterm of $\code_0$; 
      \item \emph{Erasing Redexes}: if $\skframe = \skframetwo\cons\gcentry{\la\var\codethree}{\stacktwo}\cons\skframethree$, then $\la\var\codethree$ and any code in $\stacktwo$ are subterms of $\code_0$;    
      \end{enumerate}
  \end{enumerate}
  
  \item \emph{Erasing Redexes}: \label{p:ucam-invariants-erasing} 
			  if $\skframe = \skframetwo \cons \gcentry{ \la\var\codethree }\stacktwo \cons \skframethree$ then $\var$ does not occur in $\genv$, $\var \notin \fv {\codethree}$, and $\var \notin \fv {\relunf\codethree\genv}$.
\item \emph{Decoding}: \label{p:ucam-invariants-decoding}
		  $\stctx\state$ is a \perp context.
\end{enumerate}
\end{lemma}

\begin{proof}\hfill
\begin{enumerate}
 \item \emph{Normal Form}: 
 \input{\proofspath/normal-form-invariant}
 \item\emph{Subterm}: This is a special case of the more refined subterm invariant of the \maxmam (\reflemmap{umam-quant-invs}{subterm}, page \pageref{l:umam-quant-invs}).
 \item\emph{Erasing Redexes}:
 \input{\proofspath/checks-frame-and-erasing-redexes-invariant}
 \item\emph{Decoding}:
 \input{\proofspath/checks-contextual-decoding-invariant}
\end{enumerate}
\end{proof}

\subsection{Proof of the \ucam Properties Theorem (\refthmcompact{ucam-propert}, p. \pageref{thm:ucam-propert})}
\label{ssect:ucam-propert}
\begin{proof}\hfill
\input{\proofspath/checking-AM-properties}  
\end{proof}

\subsection{Useful MAM Qualitative Study (\reflemmaeq{umam-invariants}, p. \pageref{l:umam-invariants})}
\label{ssect:umam-invariants}

Four invariants are required. The \emph{normal form} and \emph{decoding invariants} are exactly those of the \ucam (and the proof for the commutative transitions is the same). The \emph{environment labels invariant} follows from the correctness of the \ucam (\refthmp{ucam-propert}{compl}). The \emph{name invariant} is used in the proof of \reflemma{one-step-weak-sim}.

\begin{lemma}[\maxmam Qualitative Invariants\withproofs{, Proof at Page \pageref{ssect:umam-invariants}}]
	\label{l:umam-invariants}  
	Let $\state = \skamstate{\skphase}{\skframe}{\codetwo}{\stack}{\genv}$ be a state reachable from an initial term $\code_0$. Then:
	\begin{enumerate}		
		
		\item \emph{Environment Labels}: \label{p:umam-invariants-env-lab}
		 $\genv$ is well-labeled.

		\item \emph{Normal Form}: \label{p:umam-invariants-nf}
		\begin{enumerate}
			\item \emph{Backtracking Code}: \label{p:umam-invariants-nf-code} 
if $\skphase = \skback$, then $\relunf\codetwo\genv$ is normal, and if $\stack$ is non-empty, then $\relunf\codetwo\genv$ is neutral; 
			
			\item \emph{Frame}: \label{p:umam-invariants-nf-frame} 
			if $\skframe = \skframetwo\cons\skap{\codethree}{\stacktwo}\cons\skframethree$, then $\relunf\codethree\genv$ is neutral.
		\end{enumerate}
	
		\item \emph{Name:} \label{p:umam-invariants-new-name}
		\begin{enumerate}
			\item \emph{Substitutions}: \label{p:umam-invariants-new-name-es}
			if $\genv = \genvtwo \cons \esub\var\code \cons \genvthree$  then $\var$ is fresh wrt $\code$ and $\genvthree$;

			\item \emph{Abstractions and Evaluation}: \label{p:umam-invariants-new-name-abs-eval}
			if $\skphase = \skeval$ and $\la\var\code$ is a subterm of $\codetwo$, $\stack$, or $\stacktwo$ (if $\skframe = \skframetwo\cons\skap{\codethree}{\stacktwo}\cons\skframethree$) or of a erasing redex item in $\skframe$ then $\var$ may occur only in $\code$;
			\item \emph{Abstractions and Backtracking}: \label{p:umam-invariants-new-name-abs-back}
			if $\skphase = \skback$ and $\la\var\code$ is a subterm of $\stack$ or $\stacktwo$ (if $\skframe = \skframetwo\cons\skap{\codethree}{\stacktwo}\cons\skframethree$) or of a erasing redex item in $\skframe$ then $\var$ may occur only in $\code$.
		\end{enumerate}		
		
		\item \emph{Erasing Redexes}: \label{p:umam-invariants-erasing} 
			if $\skframe = \skframetwo\cons\gcentry{\la\var\codethree}{\stacktwo}\cons\skframethree$, then $\var \notin \codethree$ and $\var \notin \relunf\codethree\genv$.

		\item \emph{Decoding}: \label{p:umam-invariants-decoding}
		  $\stctx\state$ is a \perp context.
	\end{enumerate}
\end{lemma}

\begin{proof}\hfill
\begin{enumerate}
 \item \emph{Environment Labels}: the only transition that extends the environment is $\tomachmtwo$, and it preserves well-labeledness by \refthmp{ucam-propert}{correc}.
 
 \item \emph{Normal Form}: for the commutative transitions the proof is exactly as for the \ucam, see \reflemmap{ucam-invariants}{nf} whose proof is at page \pageref{ssect:ucam-invariants}. For the multiplicative and exponential transitions the invariant holds trivially: the \emph{backtracking code} part because these transition cannot happen during backtracking (note that $\tomachmthree$ may happen during backtracking but it ends in a evaluating state, for which then nothing has to be proven), and the \emph{frame} part because they do not touch the head context items in the frame.

 \item \emph{Name}: 
 \input{\proofspath/name-invariant-proof}

 \item \emph{Erasing}:   the invariant trivially holds for an initial state
  $\skamstate{\skeval}{\stempty}{\code}{\stempty}{ \stempty}$.
For a non-empty evaluation sequence, for all transitions but $\tomachasubseven$ the statement follows immediately from the \ih because they all leave untouched the erasing redex items in the frame (but for $\tomachmthree$ that, however, simply removes one such entry and so trivially preserves the invariant). So consider $\mac{ \skeval }{ \skframe }{ \la\var\codethree }{ \codetwo\cons\stack }{ \genv }
             \tomachasubseven
             \mac{ \skeval }{ \skframe \cons \gcentry{ \la\var\codethree }\pi}{ \codetwo }{ \stack }{ \genv }$. By the hypothesis on the transition, we have $\var \notin \fv\codethree$. By the name invariant (\reflemmap{umam-invariants}{new-name-abs-eval}), $\var$ does not occur in $\genv$. Last, since $\var$ does not occur in $\genv$ nor in $\codethree$, clearly it does not occur in $\relunf\codethree\genv$.
 
 \item \emph{Decoding}: 
 \input{\proofspath/contextual-decoding-invariant}
\end{enumerate}
\end{proof}

We can now show how every single transition projects on the $\l$-calculus, and in particular that multiplicative transitions project to \perp $\beta$-steps.

\begin{lemma}[One-Step Weak Simulation]
  \label{l:one-step-weak-sim} 
	Let $\state$ be a reachable state.
	\begin{enumerate}
		\item \label{p:one-step-weak-sim-com} \emph{Commutative}: 
		if $\state\tomachhole{\csym_{1,2,3,4,5,6,7}}\statetwo$ then $\decode\state=\decode\statetwo$;
		\item \label{p:one-step-weak-sim-exp} \emph{Exponential}: 
		if $\state\tomachhole{\esym_{\lababs,\labredsym}}\statetwo$ then $\decode\state=\decode\statetwo$;
		\item \label{p:one-step-weak-sim-mult}\emph{Multiplicative}: if $\state\tomachhole{\msym_{1,2,3}}\statetwo$ then $\decode\state\toperp\decode\statetwo$.		
	\end{enumerate}
\end{lemma}

\begin{proof}\hfill
\input{\proofspath/one-step-weak-simulation-proof}
\end{proof}

We also need to show that the \maxmam computes $\beta$-normal forms. 

\begin{lemma}[Progress]
\label{l:eglamour-progress}
  Let $\state$ be a reachable final state. Then $\decode\state$ is $\beta$-normal. 
\end{lemma}

\begin{proof}
A simple inspection of the machine transitions shows that final states have the form  $\mac{ \skback }{ \stempty }{ \code }{ \stempty }{ \genv }$. Then by the normal form invariant (\reflemmap{umam-invariants}{nf-code}) $\decode\state = \relunf\code\genv$ is $\beta$-normal.
\end{proof}
 
\subsection{Proof of the Weak Bisimulation Theorem (\refthmcompact{weak-bis}, p. \pageref{thm:weak-bis})}
\label{ssect:weak-bis}
 \begin{proof}\hfill
  \begin{enumerate}
   \item By induction on the length $\size\exec$ of $\exec$, using the one-step weak simulation lemma (\reflemma{one-step-weak-sim}). If $\exec$ is empty then the empty derivation satisfies the statement. If $\exec$ is given by $\exectwo:\state\tomach^*\statethree$ followed by $\statethree\tomach\statetwo$ then by \ih there exists $\derivtwo:\decode\state\toperp^*\decode\statethree$ s.t. $\size\derivtwo = \sizem\exectwo$. Cases of $\statethree\tomach\statetwo$:
   \begin{enumerate}
    \item \emph{Commutative or Exponential}. Then $\decode\statethree = \decode\statetwo$ by \reflemmap{one-step-weak-sim}{com} and \reflemmap{one-step-weak-sim}{exp}, and the statement holds taking $\deriv \defeq \derivtwo$ because $\size\deriv = \size\derivtwo =_{\ih} \sizem\exectwo = \sizem\exec$.
    
    \item \emph{Multiplicative}. Then $\decode\statethree \toperp \decode\statetwo$ by \reflemmap{one-step-weak-sim}{mult} and defining $\deriv$ as $\derivtwo$ followed by such a step we obtain $\size\deriv = \size\derivtwo +1 =_{\ih} \sizem\exectwo + 1 =  \sizem\exec$.
   \end{enumerate}
 
    \item We use $\ecnf\state$ 
  to denote the normal form of $\state$ with respect to exponential and commutative transitions, that exists and is unique because $\tomachc\cup\tomache$ terminates (termination is given by forthcoming \refthm{exp-linearity} and \refthm{commutative-bilinearity}, that are postponed because they actually give precise complexity bounds, not just termination) and the machine is deterministic (as it can be seen by an easy inspection of the transitions). The proof is by induction on the length of $\deriv$. If $\deriv$ is empty then the empty execution satisfies the statement.
  
 If $\deriv$ is given by $\derivtwo:\code\toperp^*\tmthree$ followed by $\tmthree \toperp \tmtwo$ then by \ih there is an execution $\exectwo:\state\tomach^*\statethree$ s.t. $\tmthree = \decode\statethree$ and $\sizem\exectwo = \size\derivtwo$. Note that since exponential and commutative transitions are mapped on equalities, $\exectwo$ can be extended as
  $\exectwo':\state\tomach^*\statethree\tomachhole{\eredsym,\eabssym,\csym_{1,2,3,4,5,6}}^*\ecnf{\statethree}$ with 
   $\decode{\ecnf{\statethree}} = \tmthree$ and $\sizem{\exectwo'}=\size\derivtwo$. By the progress property (\reflemma{eglamour-progress}) $\ecnf{\statethree}$ cannot be a final state, otherwise $\tmthree = \decode{\ecnf{\statethree}}$ could not reduce. Then $\ecnf{\statethree} \tomachm \statetwo$ (the transition is necessarily multiplicative because $\ecnf{\statethree}$ is normal with respect to the other transitions). By the one-step weak simulation lemma (\reflemmap{one-step-weak-sim}{mult}) $\decode{\ecnf{\statethree}} = \tmthree \toperp \decode\statetwo$ and by determinism of $\toperp$ (\reflemmap{perp-properties}{determinism}) $\decode\statetwo = \tmtwo$. Then the execution $\exec$ defined as $\exectwo'$ followed by $\ecnf{\statethree} \tomachm \statetwo$ satisfy the statement, as $\sizem{\exec} = \sizem{\exectwo'} +1 = \sizem{\exectwo} +1 = \size\derivtwo + 1 = \size\deriv$.
  \end{enumerate}
 \end{proof}

\input{Appendix_quantitative_analysis}

%% file: proofs/commutative-distallation-proof.tex
\begin{enumerate}
 \item Transitions:
\begin{itemize}
\item \case{$\mac{ \skeval }{ \skframe }{ \code\codetwo }{ \stack }{ \genv }
             \tomachasubone
             \mac{ \skeval }{ \skframe }{ \code }{ \codetwo\cons\stack }{ \genv }$}
          We have 
     $\decdump\ctxholep{\decstackp{\code\codetwo}}
     =
          \decdump\ctxholep{\decoderevp{\codetwo\cons\stack}{\code}}$.

\item \case{$\mac{ \skeval }{ \skframe }{ \la\var\code }{ \codetwo\cons\stack }{ \genv }
             \tomachasubseven
             \mac{ \skeval }{ \skframe \cons \gcentry{ \la\var\code }\pi}{ \codetwo }{ \stempty }{ \genv }$}
            We have 
          $\decdump\ctxholep{\decoderevp{\codetwo\cons\stack}{\la\var\code}}
	    = \decdump\ctxholep{\decstackp{(\la\var\code)\codetwo}}
	    = \decodep{ (\skframe \cons \gcentry{ \la\var\code }\stack)}\codetwo $. 

\item \case{$\mac{ \skeval }{ \skframe}{ \la\var\code }{ \stempty }{ \genv }
             \tomachasubtwo
             \mac{ \skeval }{ \mframe\cons\var }{ \code }{ \stempty }{ \genv }$}
    We have $\decodep{\mframe}{\la\var\code} = \decodep{\mframe\cons\var}\code$.

\item \case{$\mac{ \skeval }{ \skframe }{ \var }{ \stack }{ \genv }
             \tomachasubthree
             \mac{ \skback }{ \skframe }{ \var }{ \stack }{ \genv }$} Nothing to prove.

\item \case{$\mac{ \skback }{ \mframe\cons\var }{ \code }{ \stempty }{ \genv }
             \tomachasubfour
             \mac{ \skback }{ \skframe }{ \la\var\code }{ \stempty }{ \genv }
             $}
             Exactly as $\tomachasubtwo$.

\item \case{$\mac{ \skback }{ \skframe\cons \dentry\code\stack  }{ \codetwo }{ \stempty }{ \genv }
             \tomachasubfive
             \mac{ \skback }{ \skframe }{ \code\codetwo }{ \stack }{ \genv }$}
     We have $\decodep{ (\skframe \cons \dentry\code\stack)}\codetwo =
    \decdump\ctxholep{\decstackp{\code\codetwo}}$.

\item \case{$\mac{ \skback }{ \skframe }{ \code }{ \codetwo\cons\stack }{ \genv }
             \tomachasubsix
             \mac{ \skeval }{ \skframe \cons \dentry\code\stack }{ \codetwo }{ \stempty }{ \genv }$}
            We have 
          $\decdump\ctxholep{\decoderevp{\codetwo\cons\stack}{\code}}
	    = \decdump\ctxholep{\decstackp{\code\codetwo}}
	    = \decodep{ (\skframe \cons \dentry\code\stack)}\codetwo $.

\end{itemize}

\item We have $\decode\state = \relunf{\decodep\mframe{\decodestack{\code}{\stack}}}\genv =_{\refpointeq{comm-transp-trunk}} \relunf{\decodep{\mframe'}{\decodestack{\code'}{\stack'}}}\genv = \decode\statetwo$.\qed

\end{enumerate}

%% file: proofs/normal-form-invariant.tex
  the invariant trivially holds for an initial state
  $\skamstate{\skeval}{\stempty}{\code}{\stempty}{\genv}$.
For a non-empty evaluation sequence we list the cases for the last transition. We omit $\tomachasubone$ because it follows immediately from the \ih
  \begin{itemize}


  \item \case{$
          \mac{\skeval}{\skframe}{\la\var\code}{\stempty}{\genv}
          \tomachasubtwo
          \mac{\skeval}{\skframe\cons\var}{\code}{\stempty}{\genv}
        $}
        \begin{enumerate}
      \item \emph{Backtracking Code}: trivial since $\skphase \neq \skback$.      
      \item \emph{Frame}: it follows by the \ih, since the transition only extends $\skframe$ with an item that is not a head context.
      \end{enumerate}

\item \case{$\mac{ \skeval }{ \skframe }{ \la\var\code }{ \codetwo\cons\stack }{ \genv }
             \tomachasubseven
             \mac{ \skeval }{ \skframe \cons \gcentry{ \la\var\code }\pi}{ \codetwo }{ \stempty }{ \genv }$ with $\var \notin \fv\code$}
            \begin{enumerate}
      \item \emph{Backtracking Code}: trivial since $\skphase \neq \skback$.      
      \item \emph{Frame}: it follows by the \ih, since the transition only extends $\skframe$ with an item that is not a head context.
      \end{enumerate}

  \item \case{$
          \mac{\skeval}{\skframe}{\var}{\stack}{\genv}
          \tomachasubthree
          \mac{\skback}{\skframe}{\var}{\stack}{\genv}
        $ with $\genv(\var) = \bot$ or $\genv(\var) = \esublabe\var\codetwo\labneu$ or ($\genv(\var) = \esublabe\var\codetwo\lababs$ and $\stack = \stempty$)}
    
    \begin{enumerate}
      \item \emph{Backtracking Code}: three cases depending on $\genv$:
      \begin{enumerate}
	\item $\genv(\var) = \bot$: then $\relunf\var\genv = \var$, that is both a normal and a \quiet term.
	\item $\genv(\var) = \esublabe\var\codetwo\labneu$: more precisely $\genv = \genv_1 \cons \esublabe\var\codetwo\labneu \cons \genv_2$. Then we have $\relunf\var\genv = \relunf\var{\genv_1 \cons \esublabe\var\codetwo\labneu \cons \genv_2} = \relunf\var{\esublabe\var\codetwo\labneu \cons \genv_2} = \relunf\codetwo{\genv_2}$ that is a neutral term because $\genv$ is well-labeled. Note that $\genv_1$ cannot bound $\var$ because of the freshness requirements in the definition of well-labeled environment.
	\item $\genv(\var) = \esublabe\var\codetwo\lababs$ and $\stack = \stempty$: similarly to the previous case we obtain that $\relunf\var\genv$ is a normal abstraction.
      \end{enumerate}
      \item \emph{Frame}: it follows from the \ih, as $\skframe$ is unchanged.
      \end{enumerate}

  \item \case{$
          \mac{\skback}{\skframe\cons\var}{\code}{\stempty}{\genv}
          \tomachasubfour
          \mac{\skback}{\skframe}{\la\var\code}{\stempty}{\genv}
        $}
        	
      \begin{enumerate}
      \item \emph{Backtracking Code}: by \ih\ we know that $\relunf\code\genv$ is a normal term. Then $\relunf{(\la\var\code)}\genv = \la\var\relunf\code\genv$ is a normal term. The stack is empty, so we conclude.
      \item \emph{Frame}: it follows from the \ih, as $\skframe$ is unchanged.
      \end{enumerate}

  \item \case{$
          \mac{\skback}{ \skframe\cons\skap{\code}{\stack} }{\codetwo}{\stempty}{\genv}
          \tomachasubfive
          \mac{\skback}{\skframe}{\code\codetwo}{\stack}{\genv}
        $}
      \begin{enumerate}
      \item \emph{Backtracking Code}: by \ih\ we have that $\relunf\codetwo\genv$ is a normal term while by \refpoint{ucam-invariants-nf-frame} of the \ih $\relunf\code\genv$ is a \quiet term. Therefore $\relunf{(\code\codetwo)}\genv = \relunf\code\genv \relunf\codetwo\genv$ is a \quiet term.
      \item \emph{Frame}: it follows from the \ih, as $\skframe$ is unchanged.
      \end{enumerate}

  \item \case{$
          \mac{\skback}{\skframe}{\code}{\codetwo\cons\stack}{\genv}
          \tomachasubsix
          \mac{\skeval}{ \skframe\cons\skap{\code}{\stack} }{\codetwo}{\stempty}{\genv}
        $}
      \begin{enumerate}
      \item \emph{Backtracking Code}: trivial since $\skphase \neq \skback$.      
      \item \emph{Frame}: $\relunf\code\genv$ is a \quiet term by \refpoint{ucam-invariants-nf-code} of the \ih, the rest follows from the \ih
      \end{enumerate} 
  \end{itemize}

%% file: proofs/checks-frame-and-erasing-redexes-invariant.tex
  the invariant trivially holds for an initial state
  $\skamstate{\skeval}{\stempty}{\code}{\stempty}{ \genv}$.
For a non-empty evaluation sequence, for all transitions but $\tomachasubseven$ the statement follows immediately from the \ih because they all leave untouched the erasing redex items in the frame. So consider $\mac{ \skeval }{ \skframe }{ \la\var\codethree }{ \codetwo\cons\stack }{ \genv }
             \tomachasubseven
             \mac{ \skeval }{ \skframe \cons \gcentry{ \la\var\codethree }\pi}{ \codetwo }{ \stack }{ \genv }$. By the hypothesis on the transition, we have $\var \notin \fv\codethree$. By the subterm invariant (\reflemmap{ucam-invariants}{subterm-one}), $\la\var\codethree$ is a subterm of the initial term and so, by the hypotheses on initial states, $\var$ does not occur in $\genv$. Last, since $\var$ does not occur in $\genv$ nor in $\codethree$, clearly it does not occur in $\relunf\codethree\genv$.

%% file: proofs/checks-contextual-decoding-invariant.tex
  the invariant trivially holds for an initial state
  $\skamstate{\skeval}{\stempty}{\code}{\stempty}{ \genv}$.
For a non-empty evaluation sequence we list the cases for the last transition. To simplify the reasoning in the following case analysis we let implicit that the unfolding spreads on all the subterms, \ie that $\relunf{\decdump\ctxholep{\decodestack\code\stack}} \genv = \decodep{\relunf\mframe\genv}{\decodestack{\relunf\code\genv}{\relunf\stack\genv}}$. Cases:

  \begin{itemize}
 \item \case{$\statetwo = 
          \mac{ \skeval }{ \skframe }{ \code\codetwo }{ \stack }{ \genv }
          \tomachasubone
          \mac{ \skeval }{ \skframe }{ \code }{ \codetwo\cons\stack }{ \genv } = \state
        $}
 By \reflemmap{lo-properties}{arg}, $\stctx\state = \relunf{\decdump\ctxholep{\decodestack\cdot{\codetwo\cons\stack}}} \genv$ is \perp iff $\stctx\statetwo = \relunf{\decdump\ctxholep{\decodestack\cdot{\stack}}} \genv$ is \perp, and $\stctx\statetwo$ is \perp by \ih

  
  \item \case{$\mac{ \skeval }{ \skframe }{ \la\var\codethree }{ \codetwo\cons\stack }{ \genv }
             \tomachasubseven
             \mac{ \skeval }{ \skframe \cons \gcentry{ \la\var\codethree }\pi}{ \codetwo }{ \stempty }{ \genv }$ with $\var \notin \fv\code$}
              
             By \ih, $\stctx\statetwo = \relunf{\decode\skframe\ctxholep{\decodestack\cdot{\codetwo\cons\stack}}} \genv = \decodep{\relunf\skframe\genv}{\decodestack\cdot{\relunf\codetwo\genv\cons\relunf\stack\genv}} $ is \perp, and by \reflemmap{lo-properties}{arg} $\decodep{\relunf\skframe\genv}{\decodestack\cdot{\relunf\stack\genv}} $ is \perp. By the erasing redexes invariant (\reflemmap{ucam-invariants}{erasing}), we have $\var \notin \fv {\relunf\codethree\genv}$. 
              Then by \reflemmap{lo-properties}{erasing-red} 
             $ 
             \decodep{\relunf\skframe\genv}{ 
              \decodestack{ (\la\var\relunf\codethree\genv) \ctxhole }{\relunf\stack\genv}} 
              = \relunf{\decode\skframe\ctxholep{ \decodestack{ (\la\var\codethree) \ctxhole }{\stack} }} \genv 
              = \stctx\statetwo$ is \perp.
             
  \item \case{$\statetwo = 
          \mac{\skeval}{\skframe}{\la\var\code}{\stempty}{\genv}
          \tomachasubtwo
          \mac{\skeval}{\skframe\cons\var}{\code}{\stempty}{\genv} = \state
        $}
        We have $\stctx\state = \relunf{\decdumpp{\la\var\ctxhole}} \genv$ and by \ih we know that $\relunf{\decdumpp{\cdot}} \genv = \stctx\statetwo$ is \perp. Note that by definition the decoding of frames cannot be applicative. Then $\stctx\state $ is \perp by \reflemmap{lo-properties}{abs-app}.

  \item \case{$\statetwo = 
          \mac{\skeval}{\skframe}{\var}{\stack}{\genv}
          \tomachasubthree
          \mac{\skback}{\skframe}{\var}{\stack}{\genv} = \state
        $ with $\genv(\var) = \bot$ or $\genv(\var) = \esublabe\var\codetwo\labneu$ or ($\genv(\var) = \esublabe\var\codetwo\lababs$ and $\stack = \stempty$)}
    We have $\stctx\state = \relunf{\decdump\ctxholep{\decodestack\cdot{\stack}}} \genv = \stctx\statetwo$, and so the statement follows immediately from the \ih
      

  \item \case{$\statetwo = 
          \mac{\skback}{\skframe\cons\var}{\code}{\stempty}{\genv}
          \tomachasubfour
          \mac{\skback}{\skframe}{\la\var\code}{\stempty}{\genv} = \state
        $}
    
            We have $\stctx\state = \relunf{\decdumpp{\cdot}} \genv$ and by \ih we know that $\relunf{\decdumpp{\la\var\ctxhole}} \genv = \stctx\statetwo$ is \perp. Then $\stctx\state $ is \perp by \reflemmap{lo-properties}{abs-app}.

  \item \case{$\statetwo = 
          \mac{\skback}{ \skframe\cons\skap{\code}{\stack} }{\codetwo}{\stempty}{\genv}
          \tomachasubfive
          \mac{\skback}{\skframe}{\code\codetwo}{\stack}{\genv} = \state
        $}
By \ih $\stctx\statetwo = \relunf{\decdump\ctxholep{\decodestack{\code\ctxhole}{\stack}}} \genv$ is \perp, and so by \reflemmap{lo-properties}{neut} $\relunf{\decdump\ctxholep{\decodestack{\cdot}{\stack}}} \genv = \stctx\state$ is \perp.

  \item \case{$\statetwo = 
          \mac{\skback}{\skframe}{\code}{\codetwo\cons\stack}{\genv}
          \tomachasubsix
          \mac{\skeval}{ \skframe\cons\skap{\code}{\stack} }{\codetwo}{\stempty} {\genv} = \state
        $}
        By \ih $\stctx\statetwo = \relunf{\decdump\ctxholep{\decodestack{\cdot}{\stack}}} \genv = \decodep{\relunf\mframe\genv}{\decodestack{\cdot}{\relunf\stack\genv}} $ is \perp, and by \reflemmap{ucam-invariants}{nf-code} applied to $\statetwo$ we obtain that $\relunf\code\genv$ is neutral (because the stack is non-empty). So by \reflemmap{lo-properties}{neut} $ \decodep{\relunf\mframe\genv}{\decodestack{\relunf\code\genv\ctxhole}{\relunf\stack\genv}} = 
        \relunf{\decodep{\mframe}{\decodestack{\code\ctxhole}{\stack}}}\genv
        = \stctx\state$ is \perp.

   \end{itemize}

%% file: proofs/checking-AM-properties.tex

\begin{enumerate}
\item An inspection of the transition rules shows that there always is one and exactly one transition of the \ucam that applies: for each phase ($\skeval$ and $\skback$) consider each case of the code (application, abstraction, variable) and the various combinations stack/frame.

\item By the previous point, the executions of the \ucam are sequences of commutative steps that either diverge or are followed by an output transition. We now introduce a measure and prove that the sequence of commutative steps is always finite. In particular, it is bounded by the size of the initial term.

Consider the following notion of size for stacks, frames, and states:
\[\begin{array}{rcl@{\hspace{2em}}rcl}
\rmeasure{\stempty} & \defeq & 0 
&
\rmeasure{\mframe\cons\var} & \defeq & \rmeasure{\mframe}
\\

\rmeasure{\code\cons\stack} & \defeq & \size{\code} + \rmeasure{\stack} 
&
\rmeasure{\skframe\cons\skap{\code}{\stack}} & \defeq & \rmeasure{\stack} + \rmeasure{\skframe}\\

&&&
\rmeasure{\skframe\cons\gcentry{\la\var\code}{\stack}} & \defeq & \rmeasure{\code} +\rmeasure{\stack} + \rmeasure{\skframe}\\

 \rmeasure{\mac{\skeval}{\skframe}{\code}{\stack}{\genv}} & \defeq & \rmeasure{\skframe} + \rmeasure{\stack} + \size{\code} 
 &
\rmeasure{\mac{\skback}{\skframe}{\code}{\stack}{\genv}} & \defeq & \rmeasure{\skframe} + \rmeasure{\stack}
   \end{array}\]
The proof of the bound is in 3 steps:
\begin{enumerate}

\item \emph{Evaluation Commutative Steps Are Bound by the Size of the Initial Term}. 

By direct inspection of the rules of the machine it can be checked that:
\begin{itemize}
\item \emph{Evaluation Commutative Rules Decrease the Size}: if $\state \tomachhole{a} \statetwo$ with $a\in \set{\skeval\admsym_1,\skeval\admsym_2,\skeval\admsym_3, \skeval\admsym_7}$ then
      $\rmeasure{\statetwo} < \rmeasure{\state}$;
\item \emph{Backtracking Transitions do not Change the Size}: if $\state \tomachhole{a} \statetwo$ with
      $a\in \set{\commfour,\commfive,\commsix}$ then $\rmeasure{\statetwo} = \rmeasure{\state}$.
\end{itemize}
Then a straightforward induction on the length of an execution $\exec$ shows that
\[\rmeasure{\statetwo} \leq \rmeasure{\state} - \sizecomev\exec\]
\ie\ that $ \sizecomev\exec \leq \rmeasure{\state}  - \rmeasure{\statetwo} \leq \rmeasure{\state} = \rmeasure{\mac{\skeval}{\stempty}{\code}{\stempty}{\genv}} = \size\code$.
\item \emph{Backtracking Commutative Steps Are Bound by the Evaluation Ones}. We have to estimate $\sizecombt\exec = \polsize\exec\commfour + \polsize\exec\commfive + \polsize\exec\commsix$. Note that
\begin{enumerate}
\item $\polsize\exec\commfour \leq \polsize\exec\commtwo$, as $\tomachasubfour$ pops variables from $\skframe$, pushed only by $\tomachasubfour$;

\item $\polsize\exec\commfive \leq \polsize\exec\commsix$, as $\tomachasubfive$ pops pairs $\skap{\code}{\stack}$ from $\skframe$, pushed only by $\tomachasubsix$;

\item $\polsize\exec\commsix \leq \polsize\exec\commthree$, as $\tomachasubsix$ ends backtracking phases, started only by $\tomachasubthree$.
\end{enumerate}

Then $\sizecombt\exec \leq \polsize\exec\commtwo + 2\polsize\exec\commthree \leq 2\sizecomev\exec
$.

\item \emph{Bounding All Commutative Steps by the Size of the Intial Term}. We have $\sizecom\exec = \sizecomev\exec + \sizecombt\exec \leq_{P. 2} = \sizecomev\exec + 2\sizecomev\exec \leq_{P.1} 3\cdot\size\tm $.
\end{enumerate}

\item By the previous points executions of the \ucam are sequences of commutative transitions followed by an output transition, \ie they have the form $\state\toucmachhole{\csym}^*\statetwo\toucmachhole{\osym} \lab$ where $\lab$ is the output label of the machine. The initial state by hypothesis is $\state \defeq \mac{ \skeval }{ \stempty}{ \code }{ }{ \genv } $. Since commutative transitions do not change the decoding (\reflemma{comm-transp}), we have $\decode\statetwo = \decode\state = \relunf\code\genv$. To prove that $\esublab\var\code\cons\genv$ is well-labeled---that requires to prove properties of $\relunf\code\genv$---we will then look at $\decode\statetwo$ and at the various possible output transitions. Five cases:
 \begin{itemize}
 \item $\statetwo = \mac{ \skeval }{ \skframe }{ \la\vartwo\codethree }{ \codetwo\cons\stack }{ \genv } 
			\toucmachoutone 
			\labred 1$ with $\vartwo \in \fv\codethree$.
			Then 
			$$\relunf\code\genv = \decode\statetwo = \relunf{ \decodep\mframe{\decodestack{\la\vartwo\codethree}{\codetwo\cons\stack}} }\genv = 
			\relunf{ \decodep\mframe{\decodestack{(\la\vartwo\codethree) \codetwo}{\stack}} 
			}\genv =
			 \relunf{\decode\mframe}\genv \ctxholep{\decodestack{(\la\vartwo\relunf\codethree\genv) \relunf\codetwo\genv}{\relunf\stack\genv}} 
			$$
			that has a $\beta$-redex. 
			
			Now, we show that $\code$ decomposes as a $\rtobperp$-redex in a \perp context. By \reflemmap{ucam-invariants}{decoding}, $\relunf{ \decodep\mframe{\decodestack{\cdot}{\codetwo\cons\stack}} }\genv = \stctx\statetwo$ is a \perp context. By \reflemmap{lo-properties}{unf-rem}, $\decodep\mframe{\decodestack{\cdot}{\codetwo\cons\stack}} $ is also \perp. Then  $\decodep\mframe{\decodestack{\cdot}{\stack}} $ is \perp by \reflemmap{lo-properties}{arg}. Moreover, $(\la\vartwo\codethree)\codetwo$ is a $\rtobperp$ redex because $\vartwo \in \fv\codethree$ by hypothesis. Finally, $\code = \decodep\mframe{\decodestack{(\la\vartwo\codethree)\codetwo}{\stack}}$ by \reflemmap{comm-transp}{trunk}.
						
			The \ucam is executed only on terms that are not variables, and so $\esublabe\var\code{\labred 1}\cons\genv$ is well-labeled.
			
\item $\statetwo = \mac{ \skeval }{ \skframe }{ \vartwo }{ \stack }{ \genv } \toucmachouttwo \labred {n+1}$ with $\genv(\vartwo) = \esublabe\vartwo\codetwo{\labred n}$.
Then 
$$\relunf\code\genv = \decode\statetwo = \relunf{ \decodep\mframe{\decodestack{\vartwo}{\stack}} }\genv = 		
\relunf{\decode\mframe}\genv \ctxholep{\decodestack{\relunf\vartwo\genv}{\relunf\stack\genv}}$$
Since $\genv$ is well-labeled and $\genv(\vartwo) = \esublabe\vartwo\codetwo{\labred n}$ we have that $\relunf\vartwo\genv$ contains a $\beta$-redex, and so does $\relunf\code\genv$.

Now, we show that $\code$ decomposes as a variable in a \perp context. By \reflemmap{ucam-invariants}{decoding}, $\relunf{ \decodep\mframe{\decodestack{\cdot}{\stack}} }\genv = \stctx\statetwo$ is a \perp context. By \reflemmap{lo-properties}{unf-rem}, $\decodep\mframe{\decodestack{\cdot}{\stack}} $ is also \perp. Finally, $\code = \decodep\mframe{\decodestack{\vartwo}{\stack}}$ by \reflemmap{comm-transp}{trunk}.

Therefore $\esublabe\var\code{\labred {n+1}}\cons\genv$ is well-labeled.

\item $\statetwo = \mac{ \skeval }{ \skframe }{ \vartwo }{ \codetwo\cons\stack }{ \genv } \toucmachoutthree \labred 2$ with $\genv(\vartwo) = \esublabe\vartwo\codetwo\lababs$.
Then 
$$\relunf\code\genv = \decode\statetwo = \relunf{ \decodep\mframe{\decodestack{\vartwo}{\codetwo\cons\stack}} }\genv = 	
\relunf{ \decodep\mframe{\decodestack{\vartwo \codetwo}{\stack}} }\genv =
\relunf{\decode\mframe}\genv \ctxholep{\decodestack{\relunf\vartwo\genv \relunf\codetwo\genv}{\relunf\stack\genv}}$$
Since $\genv$ is well-labeled and $\genv(\vartwo) = \esublabe\var\codetwo\lababs$ we have that $\relunf\vartwo\genv$ is an abstraction and so $\relunf\code\genv$ contains the $\beta$-redex $\relunf\vartwo\genv \relunf\codetwo\genv$.

Now, we show that $\code$ decomposes as a variable in an applicative \perp context.  By \reflemmap{ucam-invariants}{decoding}, $\relunf{ \decodep\mframe{\decodestack{\cdot}{\codetwo\cons\stack}} }\genv = \stctx\statetwo$ is a \perp context. By \reflemmap{lo-properties}{unf-rem}, $\decodep\mframe{\decodestack{\cdot}{\codetwo\cons\stack}} $ is also \perp. Additionally, it is applicative. Finally, $\code = \decodep\mframe{\decodestack{\vartwo
}{\codetwo\cons\stack}}$ by \reflemmap{comm-transp}{trunk}.

Therefore $\esublabe\var\code{\labred 2}\cons\genv$ is well-labeled.

\item $\statetwo = \mac{ \skback }{ \stempty }{ \codetwo\codethree }{ \stempty }{ \genv } \toucmachoutfour \labneu$. Then $\code = \codetwo\codethree$ by the commutative transparency lemma (\reflemmap{comm-transp}{trunk}), and so $\code$ is an application. By the normal form invariant (\reflemmap{ucam-invariants}{nf-code}), $\relunf\code\genv = \relunf{(\codetwo\codethree)}\genv$ is normal. Moreover, it is an application, because $\relunf{(\codetwo\codethree)}\genv = \relunf{\codetwo}\genv \relunf{\codethree}\genv$, and so it is neutral. Then $\esublabe\var\code\labneu\cons\genv$ is well-labeled.

\item $\statetwo = \mac{ \skback }{ \stempty }{ \la\vartwo\codetwo }{ \stempty }{ \genv } \toucmachoutfive \lababs$. Then $\code = \la\vartwo\codetwo$ by the commutative transparency lemma (\reflemmap{comm-transp}{trunk}), and so $\code$ is an abstraction. By the normal form invariant (\reflemmap{ucam-invariants}{nf-code}), $\relunf\code\genv = \relunf{(\la\vartwo\codetwo)}\genv$ is normal. Moreover, it is an abstraction, because $\relunf{(\la\vartwo\codetwo)}\genv = \la\vartwo\relunf{\codetwo}\genv $. Then $\esublabe\var\code\lababs\cons\genv$ is well-labeled.

\item $\statetwo = \mac{ \skback }{ \skframe\cons \gcentry{ \la\vartwo\codethree }\stack  }{ \codetwo }{ \stempty }{ \genv } 
			\toucmachoutsix \labred 1$. Then
			
			$$\relunf\code\genv = \decode\statetwo =  
			\relunf{ \decodep\mframe{\decodestack{(\la\vartwo\codethree) \codetwo}{\stack}} 
			}\genv =
			 \relunf{\decode\mframe}\genv \ctxholep{\decodestack{(\la\vartwo\relunf\codethree\genv) \relunf\codetwo\genv}{\relunf\stack\genv}} 
			$$
						that has a $\beta$-redex. 
			
			Now, we show that $\code$ decomposes as a $\rtobperp$-redex in a \perp context. By \reflemmap{ucam-invariants}{decoding}, $\relunf{\decode\skframe\ctxholep{ \decodestack{ (\la\vartwo\codethree) \ctxhole }{\stack} }} \genv = \stctx\statetwo$ is a \perp context. By \reflemmap{lo-properties}{unf-rem}, $\decode\skframe\ctxholep{ \decodestack{ (\la\vartwo\codethree) \ctxhole }{\stack} } $ is also \perp. Then  $\decodep\mframe{\decodestack{\cdot}{\stack}} $ is \perp by \reflemmap{lo-properties}{arg}. Moreover, by the normal form invariant (\reflemmap{ucam-invariants}{nf-code}) $\codetwo$ is normal and by the erasing redexes invariant (\reflemmap{ucam-invariants}{erasing}) $\vartwo$ does not occur in $\codethree$. Then $(\la\vartwo\codethree)\codetwo$ is a $\rtobperp$ redex. Finally, $\code = \decodep\mframe{\decodestack{(\la\vartwo\codethree)\codetwo}{\stack}}$ by \reflemmap{comm-transp}{trunk}.
						
			The \ucam is executed only on terms that are not variables, and so $\esublabe\var\code{\labred 1}\cons\genv$ is well-labeled.
 \end{itemize}

\end{enumerate}

%% file: proofs/name-invariant-proof.tex
  the invariant trivially holds for an initial state
  $\skamstate{\skeval}{\stempty}{\code}{\stempty}{ \stempty}$.
For a non-empty evaluation sequence we list the cases for the last transition. 

\begin{itemize}

 \item Principal Cases:

   \begin{itemize}
    \item \case{$\statetwo = 
          \mac{ \skeval }{ \skframe }{ \l\var.\code }{ \vartwo\cons\stack }{ \genv }
          \tomachmone
          \mac{ \skeval }{ \skframe }{ \code\isub\var\vartwo }{ \stack }{ \genv } = \state
        $}

  \begin{enumerate}
	\item \emph{Substitution}: it follows from the \ih
	\item \emph{Abstractions and Evaluation}: it follows from the \ih
   \end{enumerate}
   
    \item \case{$\statetwo = 
          \mac{ \skeval }{ \skframe }{ \la\var\code }{ \codetwo\cons\stack }{ \genv }
          \tomachmtwo
          \mac{ \skeval }{ \skframe }{ \code }{ \stack }{ \esublab{\var}{\codetwo}\cons\genv } = \state
        $ with $\codetwo$ not a variable and $\var \in  \fv \code$}
			
  \begin{enumerate}
	\item \emph{Substitution}: it follows from the \ih for abstractions in the evaluation phase (\refpoint{umam-invariants-new-name-abs-eval}).
	\item \emph{Abstractions and Evaluation}: it follows from the \ih 
   \end{enumerate}

    \item \case{$\statetwo = 
          \mac{ \skeval }{ \skframe }{ \var }{ \stack }{ \genv }
          \tomachered
          \mac{ \skeval }{ \skframe }{ \rename{\code} }{ \stack }{ \genv } = \state$ with $\genv(\var) = \esublabe\var\codetwo{\labred n}$}

  \begin{enumerate}
	\item \emph{Substitution}: it follows from the \ih 
	\item \emph{Abstractions and Evaluation}: it follows by the \ih and the fact that in $\rename{\code}$ the abstracted variables are renamed (wrt $\code$) with fresh names.
   \end{enumerate}
          
    \item \case{$\statetwo = 
          \mac{ \skeval }{ \skframe }{ \var }{ \codetwo\cons\stack }{ \genv }
          \tomacheabs
          \mac{ \skeval }{ \skframe }{ \rename{\code} }{ \codetwo\cons\stack }{ \genv } = \state$ with $\genv(\var) = \esublabe\var\codetwo\lababs$}

   \begin{enumerate}
	\item \emph{Substitution}: it follows from the \ih 
	\item \emph{Abstractions and Evaluation}: it follows by the \ih and the fact that in $\rename{\code}$ the abstracted variables are renamed (wrt $\code$) with fresh names.	
   \end{enumerate}

	\item \case{$	\mac{ \skback }{ \skframe\cons \gcentry{ \la\var\code }\stack  }{ \codetwo }{ \stempty }{ \genv } 
			\tomachmthree 
			\mac{ \skeval }{ \skframe }{ \code }{ \stack }{ \genv }$}
   \begin{enumerate}
	\item \emph{Substitution}: it follows from the \ih 
	\item \emph{Abstractions and Evaluation}: it follows by the \ih	
   \end{enumerate}

          \end{itemize}

 \item Commutative Cases:
  \begin{itemize}
 \item \case{$\statetwo = 
          \mac{ \skeval }{ \skframe }{ \code\codetwo }{ \stack }{ \genv }
          \tomachasubone
          \mac{ \skeval }{ \skframe }{ \code }{ \codetwo\cons\stack }{ \genv } = \state
        $}
 
   \begin{enumerate}
	\item \emph{Substitution}: it follows from the \ih 
	\item \emph{Abstractions and Evaluation}: it follows from the \ih 	
   \end{enumerate}
  
    \item \case{$\mac{ \skeval }{ \skframe }{ \la\var\codethree }{ \codetwo\cons\stack }{ \genv }
             \tomachasubseven
             \mac{ \skeval }{ \skframe \cons \gcentry{ \la\var\codethree }\pi}{ \codetwo }{ \stempty }{ \genv }$ with $\var \notin \fv\code$}

	\begin{enumerate}
	\item \emph{Substitution}: it follows from the \ih 
	\item \emph{Abstractions and Evaluation}: it follows from the \ih 	
   \end{enumerate}
   
  \item \case{$\statetwo = 
          \mac{\skeval}{\skframe}{\la\var\code}{\stempty}{\genv}
          \tomachasubtwo
          \mac{\skeval}{\skframe\cons\var}{\code}{\stempty}{\genv} = \state
        $}
        
  \begin{enumerate}
	\item \emph{Substitution}: it follows from the \ih 
	\item \emph{Abstractions and Evaluation}: it follows from the \ih 	
   \end{enumerate}

  \item \case{$\statetwo = 
          \mac{\skeval}{\skframe}{\var}{\stack}{\genv}
          \tomachasubthree
          \mac{\skback}{\skframe}{\var}{\stack}{\genv} = \state
        $ with $\genv(\var) = \bot$ or $\genv(\var) = \esublabe\var\codetwo\labneu$ or ($\genv(\var) = \esublabe\var\codetwo\lababs$ and $\stack = \stempty$)
        }
          
  \begin{enumerate}
	\item \emph{Substitution}: it follows from the \ih 
	\item[3.] \emph{Abstractions and Backtracking}:	it follows from the \ih of Abstraction and Evaluation (\refpoint{umam-invariants-new-name-abs-eval}).
   \end{enumerate}

  \item \case{$\statetwo = 
          \mac{\skback}{\skframe\cons\var}{\code}{\stempty}{\genv}
          \tomachasubfour
          \mac{\skback}{\skframe}{\la\var\code}{\stempty}{\genv} = \state
        $}

  \begin{enumerate}
	\item \emph{Substitution}: it follows from the \ih 
	\item[3.] \emph{Abstractions and Backtracking}: it follows from the \ih 
   \end{enumerate}

  \item \case{$\statetwo = 
          \mac{\skback}{ \skframe\cons\skap{\code}{\stack} }{\codetwo}{\stempty}{\genv}
          \tomachasubfive
          \mac{\skback}{\skframe}{\code\codetwo}{\stack}{\genv} = \state
        $}

        \begin{enumerate}
	\item \emph{Substitution}: it follows from the \ih 
	\item[3.] \emph{Abstractions and Backtracking}:	it follows from the \ih 
   \end{enumerate}

  \item \case{$\statetwo = 
          \mac{\skback}{\skframe}{\code}{\codetwo\cons\stack}{\genv}
          \tomachasubsix
          \mac{\skeval}{ \skframe\cons\skap{\code}{\stack} }{\codetwo}{\stempty} {\genv} = \state
        $}
  \begin{enumerate}
	\item \emph{Substitution}: it follows from the \ih 
	\item \emph{Abstractions and Evaluation}: it follows from the \ih for Abstractions and Backtracking (\refpoint{umam-invariants-new-name-abs-back})
   \end{enumerate}
        
   \end{itemize}
  \end{itemize}

%% file: proofs/contextual-decoding-invariant.tex
  the invariant trivially holds for an initial state
  $\skamstate{\skeval}{\stempty}{\code}{\stempty}{ \stempty}$.
For a non-empty evaluation sequence we list the cases for the last transition. To simplify the reasoning in the following case analysis we let implicit that the unfolding spreads on all the subterms, \ie that $\relunf{\decdump\ctxholep{\decodestack\code\stack}} \genv = \decodep{\relunf\mframe\genv}{\decodestack{\relunf\code\genv}{\relunf\stack\genv}}$.

\begin{itemize}

 \item Principal Cases:

   \begin{itemize}
    \item \case{$\statetwo = 
          \mac{ \skeval }{ \skframe }{ \l\var.\code }{ \vartwo\cons\stack }{ \genv }
          \tomachmone
          \mac{ \skeval }{ \skframe }{ \code\isub\var\vartwo }{ \stack }{ \genv } = \state
        $}

By \reflemmap{lo-properties}{arg}, $\stctx\state = \relunf{\decdump\ctxholep{\decodestack\cdot{\stack}}} \genv$ is \perp iff $\stctx\statetwo = \relunf{\decdump\ctxholep{\decodestack\cdot{\vartwo\cons\stack}}} \genv$ is \perp, and $\stctx\statetwo$ is \perp by \ih

    \item \case{$\statetwo = 
          \mac{ \skeval }{ \skframe }{ \la\var \code }{ \codetwo\cons\stack }{ \genv }
          \tomachmtwo
          \mac{ \skeval }{ \skframe }{ \code }{ \stack }{ \esublab{\var}{\codetwo}\cons\genv } = \state
        $ with $\var \in \fv \code$}
	
	We have to prove that $\stctx\statetwo = \relunf{\decdump\ctxholep{\decodestack\cdot{\stack}}}{\esublab{\var}{\codetwo}\cons\genv}$ is \perp. 
By the name invariant for abstractions (\reflemmap{umam-invariants}{new-name-abs-eval}) we have that $\var$ does not occur in $\mframe$ nor $\stack$, and so $\relunf{\decdump\ctxholep{\decodestack\cdot{\stack}}}{\esublab{\var}{\codetwo}\cons\genv} = \relunf{\decdump\ctxholep{\decodestack\cdot{\stack}}}{\genv}$. Now, by \reflemmap{lo-properties}{arg} $\relunf{\decdump\ctxholep{\decodestack\cdot{\stack}}}{\genv}$ is \perp iff $ \relunf{\decdump\ctxholep{\decodestack\cdot{\codetwo\cons\stack}}} \genv$ is \perp, that holds by \ih, because this is exactly $\stctx\statetwo$.

    \item \case{$\statetwo = 
          \mac{ \skeval }{ \skframe }{ \var }{ \stack }{ \genv }
          \tomachered
          \mac{ \skeval }{ \skframe }{ \rename{\code} }{ \stack }{ \genv } = \state$ with $\genv(\var) = \esublabe\var\codetwo{\labred n}$}
          We have $\stctx\state = \relunf{\decdump\ctxholep{\decodestack\cdot{\stack}}} \genv = \stctx\statetwo$, and so the statement follows immediately from the \ih
        
    \item \case{$\statetwo = 
          \mac{ \skeval }{ \skframe }{ \var }{ \codetwo\cons\stack }{ \genv }
          \tomacheabs
          \mac{ \skeval }{ \skframe }{ \rename{\code} }{ \codetwo\cons\stack }{ \genv } = \state$ with $\genv(\var) = \esublabe\var\codetwo\lababs$}
          We have $\stctx\state = \relunf{\decdump\ctxholep{\decodestack\cdot{\stack}}} \genv = \stctx\statetwo$, and so the statement follows immediately from the \ih
          
          \item \case{$\statetwo = \mac{ \skback }{ \skframe\cons \gcentry{ \la\var\code }\stack  }{ \codetwo }{ \stempty }{ \genv } 
			\tomachmthree 
			\mac{ \skeval }{ \skframe }{ \code }{ \stack }{ \genv } = \state$}

We have $\stctx\state = \relunf{\decdump\ctxholep{\decodestack\cdot{\stack}}} \genv = \decodep{\relunf\mframe\genv}{\decodestack{\relunf\code\genv}{\relunf\stack\genv}}$. By the erasing redexes invariant (\reflemmap{umam-invariants}{erasing}), $\var \notin \relunf\codethree\genv$. Then, by \reflemmap{lo-properties}{erasing-red} $\stctx\state$ is \perp if and only if 
			$\decodep{\relunf\mframe\genv}{\decodestack{(\la\var\relunf\codethree\genv)\ctxhole}{\relunf\stack\genv}} = \relunf{ \decodep\mframe{\decodestack{(\la\vartwo\codethree) \ctxhole}{\stack}} 
			}\genv = \stctx\statetwo$ is \perp, that holds by \ih

  \end{itemize}


 \item Commutative Cases: exactly as in the proof of \reflemmap{ucam-invariants}{decoding}.
  \end{itemize}

%% file: proofs/one-step-weak-simulation-proof.tex

%
\begin{enumerate}

\item \emph{Commutative}: the proof is exactly as the one for the \ucam (\reflemmap{comm-transp}{state})\withproofs{, that can be found at page \pageref{ssect:ucam-invariants}}\withoutproofs{, that can be found in the longer version of this paper on the author's webpage}. 

\item \emph{Exponential}: 
\begin{itemize}
  \item \case{$\state = 
          \mac{ \skeval }{ \skframe }{ \var }{ \stack }{ \genv }
          \tomachered
          \mac{ \skeval }{ \skframe }{ \rename{\code} }{ \stack }{ \genv } = \statetwo$ with $\genv(\var) = \esublabe\var\code{\labred n}$}
        Then $\genv = \genvtwo \cons \esublabe\var\code{\labred n} \cons \genvthree$ for some environments $\genvtwo$, and $\genvthree$. Remember that terms are considered up to $\alpha$-equivalence.
        \[\begin{array}{llllllllllll}    
    \decode{\state}
    &=&
    \stctxp\statetwo{\relunf\var\genv}
    &=&
\stctxp\statetwo{\relunf\code\genvthree}
    
    &=&
\stctxp\statetwo{\relunf\code\genv}
            & = &
    \decode{\statetwo}
    \end{array}\]
        In the chain of equalities we can replace $\relunf\code\genvthree$ with $\relunf\code\genv$ because by well-labeledness the variables bound by $\genvtwo$ are fresh with respect to $\code$.
        
    \item \case{$\state = 
          \mac{ \skeval }{ \skframe }{ \var }{ \codetwo\cons\stack }{ \genv }
          \tomacheabs
          \mac{ \skeval }{ \skframe }{ \rename{\code} }{ \codetwo\cons\stack }{ \genv } = \statetwo$ with $\genv(\var) = \esublabe\var\code\lababs$}
          The proof that $\decode{\state} = \decode{\statetwo}$ is exactly as in the previous case. 
\end{itemize}

\item \emph{Multiplicative}:

   \begin{itemize}
    \item \case{$\state = 
          \mac{ \skeval }{ \skframe }{ \la\var\code }{ \vartwo\cons\stack }{ \genv }
          \tomachmone
          \mac{ \skeval }{ \skframe }{ \code\isub\var\vartwo }{ \stack }{ \genv } = \statetwo
        $} 
        Note that $\decode{\stctx{\statetwo}} = \relunf{\decdump\ctxholep{\decstack}} \genv$ is \perp by the decoding invariant (\reflemmap{umam-invariants}{decoding}). Note also that by the name invariant (\reflemmap{umam-invariants}{new-name-abs-eval}) $\var$ can only occur in $\code$. Then:
     \[\begin{array}{lllllllll}    
    \decode{\mac{ \skeval }{ \skframe }{ \la\var\code }{ \vartwo\cons\stack }{ \genv }}
    &=&
    \relunf{\decode\skframe\ctxholep{\decodestack{\l\var.\code}{\vartwo\cons\stack}}}\genv\\
    &=&
\relunf{\decode\skframe\ctxholep{\decodestack{(\l\var.\code)\vartwo}\stack }}\genv
    \\
    &=&
\stctxp\statetwo{(\l\var.\relunf\code\genv)\relunf\vartwo\genv}
    \\
    &\toperp&
\stctxp\statetwo{\relunf\code\genv \isub\var{\relunf\vartwo\genv}}
    \\
    &=_{\reflemmaeqp{umam-invariants}{new-name-abs-eval} \& \reflemmaeqp{grft-unf-proper}{two}}&
\stctxp\statetwo{\relunf{\code \isub\var\vartwo}\genv}\\

        & = &
    \decode{\mac{ \skeval }{ \skframe }{ \code\isub\var\vartwo }{ \stack }{ \genv }}
    \end{array}\]

    \item \case{$\state = 
          \mac{ \skeval }{ \skframe }{ \l\var.\code }{ \codetwo\cons\stack }{ \genv }
          \tomachmtwo
          \mac{ \skeval }{ \skframe }{ \code }{ \stack }{ \esublab{\var}{\codetwo}\cons\genv } = \statetwo
        $ with $\codetwo$ not a variable}
			By the name invariant (\reflemmap{umam-invariants}{new-name-abs-eval}) $\var$ can only occur in $\code$ and so $\decode{\stctx{\statetwo}} = \relunf{\decdump\ctxholep{\ctxhole\decstack}} {\esublab{\var}{\codetwo}\cons\genv} = \relunf{\decdump\ctxholep{\ctxhole\decstack}} \genv = \decode{\relunf\skframe\genv}\ctxholep{\decodestack{\cdot }{\relunf\stack\genv }}$. Moreover, such a context is \perp by the decoding invariant (\reflemmap{umam-invariants}{decoding}). Then:
     \[\begin{array}{lllllllll}    
    \decode{\mac{ \skeval }{ \skframe }{ \la\var\code }{ \codetwo\cons\stack }{ \genv }}
    &=&
    \relunf{\decode\skframe\ctxholep{\decodestack{\l\var.\code}{\codetwo\cons\stack}}}\genv\\
    &=&
\relunf{\decode\skframe\ctxholep{\decodestack{(\l\var.\code)\codetwo}\stack }}\genv
    \\
    &=&
    \decode{\relunf\skframe\genv}\ctxholep{\decodestack{(\la\var\relunf\code\genv) \relunf\codetwo\genv }{\relunf\stack\genv }}
    \\
    &\toperp&
\decode{\relunf\skframe\genv}\ctxholep{\decodestack{\relunf\code\genv \isub\var{\relunf\codetwo\genv} }{\relunf\stack\genv }}
    \\
    &=_{\reflemmaeqp{umam-invariants}{new-name-abs-eval} \& \reflemmaeqp{grft-unf-proper}{two}}&
\decode{\relunf\skframe\genv}\ctxholep{\decodestack{ \relunf{\code \isub\var\codetwo}\genv }{\relunf\stack\genv }}
\\
    &=&
\relunf{ \decode{\skframe}\ctxholep{\decodestack{ \code \isub\var\codetwo }{\stack }} }\genv
\\
    &=_{\reflemmaeqp{umam-invariants}{new-name-abs-eval} \& \reflemmaeqp{grft-unf-proper}{one}}&
\relunf{ \decode{\skframe}\ctxholep{\decodestack{ \code \csub\var\codetwo }{\stack }} }\genv
\\
    &=_{\reflemmaeqp{umam-invariants}{new-name-abs-eval}}&
\relunf{ \decode{\skframe}\ctxholep{\decodestack{ \code }{\stack }} \csub\var\codetwo }\genv
\\
    &=&
\relunf{ \decode{\skframe}\ctxholep{\decodestack{ \code }{\stack }} }{ \esublab\var\codetwo \cons \genv }\\

        & = &
    \decode{\mac{ \skeval }{ \skframe }{ \code }{ \stack }{ \esublab{\var}{\codetwo}\cons\genv }}
    \end{array}\]
    
    \item \case{$\state = \mac{ \skback }{ \skframe\cons \gcentry{ \la\var\code }\stack  }{ \codetwo }{ \stempty }{ \genv } 
			\tomachmthree 
			\mac{ \skeval }{ \skframe }{ \code }{ \stack }{ \genv } = \statetwo$}

Note that $\decode{\stctx{\statetwo}} =  \relunf{\decdump\ctxholep{\ctxhole\decstack}} \genv = \decode{\relunf\skframe\genv}\ctxholep{\decodestack{\cdot }{\relunf\stack\genv }}$ is \perp by the decoding invariant (\reflemmap{umam-invariants}{decoding}). Then:
     \[\begin{array}{lllllllll}    
			\decode{ \mac{ \skback }{ \skframe\cons \gcentry{ \la\var\code }\stack  }{ \codetwo }{ \stempty }{ \genv }}
			& = &
			\relunf{ \decodep\mframe{\decodestack{(\la\var\code) \codetwo}{\stack}} 
			}\genv \\
			& = &
			 \relunf{\decode\mframe}\genv \ctxholep{\decodestack{(\la\var\relunf\code\genv) \relunf\codetwo\genv}{\relunf\stack\genv}}\\
		     & \toperp &
			 \relunf{\decode\mframe}\genv \ctxholep{\decodestack{\relunf\code\genv}{\relunf\stack\genv}}\\
			& = &
			\relunf{\decdump\ctxholep{\decodestack\code\stack}} \genv
			& = &			
			\decode{ \mac{ \skeval }{ \skframe }{ \code }{ \stack }{ \genv } }
			
    \end{array}\]

  \end{itemize}
  
  \end{enumerate}

%% file: Appendix_quantitative_analysis.tex
\subsection{Quantitative Analysis}
\label{ssect:quant-analysis}

The complexity analyses of this section rely on two additional invariants of the \maxmam, the subterm and the environment size invariants.

The subterm invariant bounds the size of the duplicated subterms and it is crucial. For us, $\codetwo$ is a subterm of $\code$ if it does so up to variable names, both free and bound. More precisely: define $\tm^-$ as $\tm$ in which all variables (including those appearing in binders) are replaced by a fixed symbol $\ast$. Then, we will consider $\tmtwo$ to be a subterm of $\tm$ whenever $\tmtwo^-$ is a subterm of $\tm^-$ in the usual sense. The key property ensured by this definition is that the size $\size\codetwo$ of $\codetwo$ is bounded by $\size\code$.

\begin{lemma}[\maxmam Quantitative Invariants]
\label{l:umam-quant-invs} 
Let $\state = \skamstate{\skphase}{\skframe}{\codetwo}{\stack}{\genv}$ be a
state reachable by the execution $\exec$ from the initial code $\code_0$.

\begin{enumerate}
 \item \emph{Subterm}:
 \label{p:umam-quant-invs-subterm} 

  \begin{enumerate}
  \item \emph{Evaluating Code}: \label{p:umam-quant-invs-subterm-one}
  if $\skphase = \skeval$, then $\codetwo$ is a subterm of $\code_0$;
  
  \item \emph{Stack}:  \label{p:umam-quant-invs-subterm-two}
  any code  in the stack $\stack$ is a subterm of $\code_0$;
  
  \item \emph{Frame}: \label{p:umam-quant-invs-subterm-three}
  \begin{enumerate}
      \item \emph{Head Contexts}: if $\skframe = \skframetwo\cons\skap{\codethree}{\stacktwo}\cons\skframethree$, then any code in $\stacktwo$ is a subterm of $\code_0$; 
      \item \emph{Erasing Redexes}: if $\skframe = \skframetwo\cons\gcentry{\la\var\codethree}{\stacktwo}\cons\skframethree$, then $\la\var\codethree$ and any code in $\stacktwo$ are subterms of $\code_0$;    
      \end{enumerate}  
  
  \item \emph{Global Environment}: \label{p:umam-quant-invs-subterm-four}
  if $\genv = \genvtwo\cons\esublab{\var}{\codethree}\cons\genvthree$, then $\codethree$ is a subterm of $\code_0$;
  \end{enumerate}

  \item \emph{Environment Size}: \label{p:umam-quant-invs-env-size}
  the length of the global environment $\genv$ is bound by $\sizem\exec$.
\end{enumerate}
\end{lemma}

\begin{proof}
\input{\proofspath/subterm-invariant-proof}  
\end{proof}

The proof of the polynomial bound of the overhead is in three steps. First, we bound the number $\sizee\exec$ of exponential transitions of an execution $\exec$ using the number $\sizem\exec$ of multiplicative transitions of $\exec$, that by \refthm{weak-bis} corresponds to the number of \perp $\beta$-steps on the $\l$-calculus. Second, we bound the number $\sizecom\exec$ of commutative transitions of $\exec$ by using the number of exponential transitions and the size of the initial term. Third, we put everything together.

\paragraph{Multiplicative vs Exponential Analysis} This step requires two auxiliary lemmas. The first one essentially states that commutative transitions \emph{eat} normal and neutral terms, as well as \perp contexts.

\begin{lemma}
\label{l:comm-and-lo-ctxs} 
Let $\state = \skamstate{\skeval}{\skframe}{\code}{\stack}{\genv}$ be a state and $\genv$ be well-labeled. Then

\begin{enumerate}
\item \label{p:comm-and-lo-ctxs-one}
If $\relunf\code\genv$ is a normal term and $\pi = \stempty$ then $\state \tomachhole{\csym}^* \skamstate{\skback}{\skframe}{\code}{\stack}{\genv}$.

\item \label{p:comm-and-lo-ctxs-two}
If $\relunf\code\genv$ is a neutral term then $\state \tomachhole{\csym}^* \skamstate{\skback}{\skframe}{\code}{\stack}{\genv}$.

\item \label{p:comm-and-lo-ctxs-three}
If $\code = \ctxp\codetwo$ with $\relunf\ctx\genv$ a \perp context then there exist $\skframetwo$ and $\stacktwo$ such that $\state \tomachc^* \skamstate{\skeval}{\skframetwo}{\codetwo}{\stacktwo}{\genv}$. Moreover, if $\ctx$ is applicative then $\stacktwo$ is non-empty.

\end{enumerate}
\end{lemma}

\input{\proofspath/commutatives-and-lo-contexts}

The second lemma uses \reflemma{comm-and-lo-ctxs} and the environment labels invariant (\reflemmap{umam-invariants}{env-lab}) to show that the exponential transitions of the \maxmam are indeed useful, as they head towards a multiplicative transition, that is towards $\beta$-redexes.

\begin{lemma}[Useful Exponentials Lead to Multiplicatives]
\label{l:exp-decr} 
Let $\state$ be a reachable state such that $\state \tomacheredp n \statetwo$.
\begin{enumerate}
\item \label{p:exp-decr-one} If $n = 1$ then $\statetwo\tomachc^*\tomachm\statethree$;
\item \label{p:exp-decr-two} If $n = 2$ then $\statetwo\tomachc^*\tomacheabs\tomachc^*\tomachm\statethree$ or $\statetwo\tomachc^*\tomacheredp 1\statethree$;
\item \label{p:exp-decr-three} If $n > 1$ then $\statetwo\tomachc^*\tomacheredp{n-1}\statethree$.
\end{enumerate}
\end{lemma}

\begin{proof}
\input{\proofspath/exponentials-decrease-proof}  
\end{proof}

Finally, using the environment size invariant (\reflemmap{umam-quant-invs}{env-size}) we obtain the local boundedness property, that is used to infer a quadratic bound via a standard reasoning (already employed in \cite{DBLP:journals/corr/AccattoliL16}).

\begin{theorem}[Exponentials vs Multiplicatives]
\label{thm:exp-linearity} 
  Let  $\state$ be an initial \maxmam state and $\exec: \state \tomach^* \statetwo$. 
  \begin{enumerate}
   \item \label{p:exp-linearity-local-bound}
   \emph{Local Boundedness}: if $\exectwo:\statetwo \tomach^*\statethree$ and $\sizem\exectwo = 0$ then $\sizee\exectwo \leq \sizem\exec$;
   \item \label{p:exp-linearity-quad}
   \emph{Exponentials are Quadratic in the Multiplicatives}: $\sizee \exec \in O(\sizem\exec^2)$.
  \end{enumerate}
\end{theorem}

\begin{remark}
  \label{rem:well-lab-length}
 In the following proof we use the fact that in from the definition of well-labeled environment (\refdef{well-lab-env}, page \pageref{def:well-lab-env}) it immediately follows that if $\genv = \genvtwo \cons \esublabe\var\code{\labred n} \cons \genvthree$ is well-labeled then the length of $\genvthree$, and thus of $\genv$, is at least $n$.
\end{remark}

\begin{proof}
 \hfill
  \begin{enumerate}
   \item We prove that $\sizee\exectwo \leq \size\genv$. The statement follows from the environment size invariant (\reflemmap{umam-quant-invs}{env-size}), for which $\size\genv \leq \sizem\exec$.
   
   If $\sizee\exectwo = 0$ it is immediate. Then assume $\sizee\exectwo > 0$, so that there is a first exponential transition in $\exectwo$, \ie $\exectwo$ has a prefix $\statetwo\tomachc^*\tomache\statefour$ followed by an execution $\execthree:\statefour \tomach^* \statethree$ such that $\sizem\execthree = 0$. Cases of the first exponential transition $\tomache$:
   \begin{itemize}
    \item  Case $\tomacheabs$: the next transition is necessarily multiplicative, and so $\execthree$ is empty. Then $\sizee\exectwo = 1$. Since the environment is non-empty (otherwise $\tomacheabs$ could not apply), $\sizee\exectwo \leq \size\genv$ holds.
    
    \item Case $\tomacheredp n$: we prove by induction on $n$ that $\sizee\exectwo \leq n$, that gives what we want because $n\leq \size\genv$ by \refrem{well-lab-length}. Cases:
    \begin{itemize}
     \item $n = 1$) Then $\execthree$  has the form $\statefour \tomachc^* \statethree$ by \reflemmap{exp-decr}{one}, and so $\sizee\exectwo = 1$.
     
     \item $n = 2$) Then $\execthree$ is a prefix of $\tomachc^*\tomacheabs \tomachc^*$ or $\tomachc^*\tomacheredp 1$ by \reflemmap{exp-decr}{two}. In both cases $\sizee\exectwo \leq 2$.
     
     \item $n > 2$) Then by \reflemmap{exp-decr}{three} $\execthree$  is either  shorter or equal to $\tomachc^*\tomacheredp{n-1}$, and so $\sizee\exectwo \leq 2$, or it is longer than $\tomachc^*\tomacheredp{n-1}$, \ie it writes as $\tomachc^*$ followed by an execution $\execthree'$ starting with $\tomacheredp{n-1}$. By \ih $\size{\execthree'}\leq n-1$ and so $\size\exectwo \leq n$.
    \end{itemize}   
   \end{itemize}
   
   \item This is a standard reasoning: since by local boundedness (the previous point) $\msym$-free sequences have a number of $\esym$-transitions that are bound by the number of preceding $\msym$-transitions, the sum of all $\esym$-transitions is bound by the square of $\msym$-transitions. It is analogous to the proof of Theorem 7.2.3 in \cite{DBLP:journals/corr/AccattoliL16}. 
  \end{enumerate}
 
 \end{proof}

\paragraph{Commutative vs Exponential Analysis} The second step is to bound the number of commutative transitions. Since the commutative part of the \maxmam is essentially the same as the commutative part of the Strong MAM of \cite{DBLP:conf/aplas/AccattoliBM15}, the proof of such bound is essentially the same as in \cite{DBLP:conf/aplas/AccattoliBM15}. It relies on the subterm invariant (\reflemmap{umam-quant-invs}{subterm}).

\begin{theorem}[Commutatives vs Exponentials]
\label{thm:commutative-bilinearity} 
Let $\exec:\state\tomach^*\statetwo$ be a \maxmam execution from an initial state of code $\tm$. Then:
\begin{enumerate}
 \item \emph{Commutative Evaluation Steps are Bilinear}: $\sizecomev\exec \leq (1+\sizee\exec)\cdot\size\tm$.
 \item \emph{Commutative Evaluation Bounds Backtracking}: $\sizecombt\exec \leq 2\cdot\sizecomev\exec$.
 \item \label{p:commutative-bilinearity-three}
 \emph{Commutative Transitions are Bilinear}: $\sizecom\exec \leq 3\cdot(1+\sizee\exec)\cdot\size\tm$.
\end{enumerate}
\end{theorem}

\begin{proof}
\hfill
\input{\proofspath/bilinearity_proof}
\end{proof}

\paragraph{The Main Theorem} Putting together the matching between \perp $\beta$-steps and multiplicative transitions (\refthm{weak-bis}), the quadratic bound on the exponentials via the multiplicatives (\refthmp{exp-linearity}{quad}) and the bilinear bound on the commutatives (\refthmp{commutative-bilinearity}{three}) we obtain that the number of the \maxmam transitions to implement a \perp $\beta$-derivation $\deriv$ is at most quadratic in the length of $\deriv$ and linear in the size of $\tm$. Moreover, the subterm invariant (\reflemmap{umam-quant-invs}{subterm}) and the analysis of the \ucam (\refthmp{ucam-propert}{compl}) allow to bound the cost of implementing the execution on RAM.

\begin{theorem*}[\maxmam Overhead Bound]
\label{thm:final-thm}
  Let $\deriv:\tm \toperp^* \tmtwo$ be a maximal derivation and $\exec$ be the \maxmam execution simulating $\deriv$ given by \refthmp{weak-bis}{rev-sim}. Then:
  \begin{enumerate}
   \item \emph{Length}: $\size\exec = O((1+\size{\deriv}^2)\cdot{\size\tm})$.
   
   \item \emph{Cost}: $\exec$ is implementable on RAM in $O((1+\size{\deriv}^2)\cdot{\size\tm})$ steps.
  \end{enumerate}
\end{theorem*}

\begin{proof}\hfill
 \begin{enumerate}
  \item By definition, the length of the execution $\exec$ simulating $\deriv$ is given by $\size\exec = \sizem\exec + \sizee\exec +\sizecom\exec$. Now, by \refthmp{exp-linearity}{quad} we have $\sizee\exec = O(\sizem\exec^2)$ and by \refthmp{commutative-bilinearity}{three} we have $\sizecom\exec = O((1+\sizee\exec)\cdot\size\tm) = O((1+\sizem\exec^2)\cdot\size\tm)$. Therefore, $\size\exec = O((1+\sizee\exec)\cdot\size\tm) = O((1+\sizem\exec^2)\cdot\size\tm)$. By \refthmp{weak-bis}{rev-sim} $\sizem\exec = \size\deriv$, and so $\size\exec =  O((1+\size\deriv^2)\cdot\size\tm)$.
  
  \item The cost of implementing $\exec$ is the sum of the costs of implementing the multiplicative, exponential, and commutative transitions. Remember that the idea is that variables are implemented as references, so that environment can be accessed in constant time (\ie they do not need to be accessed sequentially). Moreover, we assuming the strong hypothesis on the representation of terms explained in the paragraph before \refthm{ucam-propert} (page \pageref{thm:ucam-propert}), so that the tests $\var \in \fv\code$ and $\var \notin \fv\code$ in $\tomachmtwo$ and  $\tomachasubseven$ can be done in constant time:   
   \begin{enumerate}
     \item \emph{Commutative}:  every commutative transition takes constant time. At the previous point we bounded their number with $O((1+\size\deriv^2)\cdot\size\tm)$, which is then also the cost of all the commutative transitions together.
     \item \emph{Multiplicative}: 
     \begin{itemize}
       \item a $\tomachmone$ transition costs $O(\size\tm)$ because it requires to rename the current code, whose size is bound by the size of the initial term by the subterm invariant (\reflemmap{umam-quant-invs}{subterm-one}). 
       
       \item A $\tomachmtwo$ transition costs $O(\size\tm)$ because checking $\var \in \fv\code$ can be done in constant time and executing the \ucam on $\codetwo$ takes $O(\size\codetwo)$ commutative steps (\refthmp{ucam-propert}{compl}), commutative steps take constant time, and the size of $\codetwo$ is bound by $\size\tm$ by the subterm invariant (\reflemmap{umam-quant-invs}{subterm-two}).
       
       \item A $\tomachmtwo$ transition takes constant time.
     \end{itemize}
Therefore, all together the multiplicative transitions cost $O(\size\deriv\cdot\size\tm)$.

\item \emph{Exponential}: At the previous point we bounded their number with $\sizee\exec = O(\size\deriv^2)$. Each exponential step copies a term from the environment, that by the subterm invariant (\reflemmap{umam-quant-invs}{subterm-four}) costs at most $O(\size\tm)$, and so their full cost is $O((1+\size\deriv)\cdot{\size\tm}^2)$ (note that this is exactly the cost of the commutative transitions, but it is obtained in a different way).
     
   \end{enumerate}
   Then implementing $\exec$ on RAM takes $O((1+\size\deriv)\cdot{\size\tm}^2)$ steps.
  \end{enumerate}
  \end{proof}

\begin{remark}
 \label{rem:linear-bound}
 Our bound is quadratic in the number of the \perp $\beta$-steps but we believe that it is not tight. In fact, our transition $\tomachmone$ is a standard optimisation, used for instance in Wand's \cite{DBLP:journals/lisp/Wand07} (section 2), Friedman et al.'s \cite{DBLP:journals/lisp/FriedmanGSW07} (section 4), and Sestoft's \cite{DBLP:journals/jfp/Sestoft97} (section 4), and motivated as an optimization about \emph{space}. In Sands, Gustavsson, and Moran's \cite{DBLP:conf/birthday/SandsGM02}, however, it is shown that it lowers the overhead for \emph{time} from quadratic to linear (with respect to the number of $\beta$-steps) for call-by-name evaluation in a weak setting. Unfortunately, the simple proof used in \cite{DBLP:conf/birthday/SandsGM02} does not scale up to our setting, nor we have an alternative proof that the overhead is linear. We conjecture, however, that it does.
\end{remark}

%% file: proofs/subterm-invariant-proof.tex

\begin{enumerate}
 \item \emph{Subterm}: the invariant trivially holds for the initial state
  $\skamstate{\skeval}{\stempty}{\code_0}{\stempty}{\stempty}$.
  In the inductive case we look at the last transition:

\begin{itemize}

 \item Principal Cases:

   \begin{itemize}
    \item \case{$\statetwo = 
          \mac{ \skeval }{ \skframe }{ \l\var.\code }{ \vartwo\cons\stack }{ \genv }
          \tomachmone
          \mac{ \skeval }{ \skframe }{ \code\isub\var\vartwo }{ \stack }{ \genv } = \state
        $}

        \begin{enumerate}
      \item \emph{Evaluating Code}: note that according to our definition of subterm $\code\isub\var\vartwo$ is a subterm of $\l\var.\code$.
      \item \emph{Stack}: note that any piece code in $\stack$ is also in $\codetwo\cons\stack$.
      \item \emph{Frame}: it follows from the \ih, since $\skframe$ is not modified.
      \item \emph{Environment}: it follows from the \ih, since $\genv$ is not modified.
      \end{enumerate}

    \item \case{$\statetwo = 
          \mac{ \skeval }{ \skframe }{ \l\var.\code }{ \codetwo\cons\stack }{ \genv }
          \tomachmtwo
          \mac{ \skeval }{ \skframe }{ \code }{ \stack }{ \esublab{\var}{\codetwo}\cons\genv } = \state
        $ with $\codetwo$ not a variable}
			
\begin{enumerate}
      \item \emph{Evaluating Code}: note that $\code$ is a subterm of $\l\var.\code$.
      \item \emph{Stack}: note that any piece code in $\stack$ is also in $\codetwo\cons\stack$.
      \item \emph{Frame}: it follows from the \ih, since $\skframe$ is not modified.
      \item \emph{Environment}: the new environment is of the form $\esublab{\var}{\codetwo}\cons\genv$.
            Pieces of code in $\genv$ are subterms of $\code_0$ by \ih.
            Moreover $\codetwo$ is the top of the stack $\codetwo\cons\stack$ so it is also
            a subterm of $\code_0$.
      \end{enumerate} 
			
    \item \case{$\statetwo = 
          \mac{ \skeval }{ \skframe }{ \var }{ \stack }{ \genv }
          \tomachered
          \mac{ \skeval }{ \skframe }{ \rename{\code} }{ \stack }{ \genv } = \state$ with $\genv(\var) = \esublabe\var\codetwo{\labred n}$}

          \begin{enumerate}
      \item \emph{Evaluating Code}: note that $\code$ is bound by $\genv$. By \ih, it is a subterm of $\code_0$.
            So $\rename{\code}$ is also a subterm of $\code_0$.
      \item \emph{Stack}: it follows from the \ih, since the stack $\stack$ is unchanged.
      \item \emph{Frame}: it follows from the \ih, since the frame $\skframe$ is unchanged.
      \item \emph{Environment}: it follows from the \ih, since the environment $\genv$ is unchanged.
      \end{enumerate} 
      
    \item \case{$\statetwo = 
          \mac{ \skeval }{ \skframe }{ \var }{ \codetwo\cons\stack }{ \genv }
          \tomacheabs
          \mac{ \skeval }{ \skframe }{ \rename{\code} }{ \codetwo\cons\stack }{ \genv } = \state$ with $\genv(\var) = \esublabe\var\codetwo\lababs$}
          
          \begin{enumerate}
      \item \emph{Evaluating Code}: note that $\code$ is bound by $\genv$. By \ih, it is a subterm of $\code_0$.
            So $\rename{\code}$ is also a subterm of $\code_0$.
      \item \emph{Stack}: it follows from the \ih, since the stack $\codetwo\cons\stack$ is unchanged.
      \item \emph{Frame}: it follows from the \ih, since the frame $\skframe$ is unchanged.
      \item \emph{Environment}: it follows from the \ih, since the environment $\genv$ is unchanged.
      \end{enumerate}

   \item \case{$\state = \mac{ \skback }{ \skframe\cons \gcentry{ \la\var\code }\stack  }{ \codetwo }{ \stempty }{ \genv } 
			\tomachmthree 
			\mac{ \skeval }{ \skframe }{ \code }{ \stack }{ \genv } = \statetwo$}
    \begin{enumerate}
      \item \emph{Evaluating Code}: by the \emph{frame} part of the \ih $\la\var\code$ is a subterm of $\code_0$, and so is $\code$.
            So $\rename{\code}$ is also a subterm of $\code_0$.
      \item \emph{Stack}: by the \emph{frame} part of the \ih the codes in $\stack$ are subterms of $\code_0$.
      \item \emph{Frame}: it follows from the \ih
      \item \emph{Environment}: it follows from the \ih, since the environment $\genv$ is unchanged.
      \end{enumerate} 
  \end{itemize}

 \item Commutative Cases:
  \begin{itemize}
 \item \case{$\statetwo = 
          \mac{ \skeval }{ \skframe }{ \code\codetwo }{ \stack }{ \genv }
          \tomachasubone
          \mac{ \skeval }{ \skframe }{ \code }{ \codetwo\cons\stack }{ \genv } = \state
        $}
 
\begin{enumerate}
      \item \emph{Evaluating Code}: by \ih, $\code\codetwo$ is a subterm of $\code_0$, so $\code$ is also a subterm of $\code_0$.
      \item \emph{Stack}: by \ih, $\code\codetwo$ is a subterm of $\code_0$, so $\codetwo$ is also a subterm of $\code_0$.
            Moreover, any piece of code in $\stack$ is a subterm of $\code_0$ by \ih.
      \item \emph{Frame}: it follows from the \ih, since the frame $\skframe$ is unchanged.
      \item \emph{Environment}: it follows from the \ih, since the environment $\genv$ is unchanged.
      \end{enumerate}

\item \case{$\statetwo = \mac{ \skeval }{ \skframe }{ \la\var\code }{ \codetwo\cons\stack }{ \genv }
             \tomachasubseven
             \mac{ \skeval }{ \skframe \cons \gcentry{ \la\var\code }\pi}{ \codetwo }{ \stempty }{ \genv } = \state$}
             
      \begin{enumerate}
      \item \emph{Evaluating Code}: by \ih, $\la\var\code$ is a subterm of $\code_0$, so $\code$ is also a subterm of $\code_0$.
      \item \emph{Stack}: nothing to prove.
      \item \emph{Frame}: for all entries in $\skframe$ it follows from the \ih, for the new entry it follows by the \emph{evaluating code} and \emph{stack} part of the \ih
      \item \emph{Environment}: it follows from the \ih, since the environment $\genv$ is unchanged.
      \end{enumerate} 
      
      \item \case{$\statetwo = 
          \mac{\skeval}{\skframe}{\la\var\code}{\stempty}{\genv}
          \tomachasubtwo
          \mac{\skeval}{\skframe\cons\var}{\code}{\stempty}{\genv} = \state
        $}
        
\begin{enumerate}
      \item \emph{Evaluating Code}: note that $\code$ is a subterm of $\l\var.\code$ which is in turn a subterm of $\code_0$ by \ih.
      \item \emph{Stack}: trivial since the stack $\stack$ is empty.
      \item \emph{Frame}: any pair of the form $\skap{\codetwo}{\stacktwo}$ or $\gcentry{\la\var\codetwo}{\stacktwo}$ in the frame $\var\cons\skframe$ is also
            already present in $\skframe$, so it follows by the \ih
      \item \emph{Environment}: it follows from the \ih, since the environment $\genv$ is unchanged.
      \end{enumerate}

  \item \case{$\statetwo = 
          \mac{\skeval}{\skframe}{\var}{\stack}{\genv}
          \tomachasubthree
          \mac{\skback}{\skframe}{\var}{\stack}{\genv} = \state
        $ with $\genv(\var) = \bot$ or $\genv(\var) = \esublabe\var\codetwo\labneu$ or ($\genv(\var) = \esublabe\var\codetwo\lababs$ and $\stack = \stempty$)
        }
    
    \begin{enumerate}
      \item \emph{Evaluating Code}: trivial since $\skphase \neq \skeval$.
      \item \emph{Stack}: it follows from the \ih, since the stack $\stack$ is unchanged.
      \item \emph{Frame}: it follows from the \ih, since the frame $\skframe$ is unchanged.
      \item \emph{Environment}: it follows from the \ih, since the environment $\genv$ is unchanged.
      \end{enumerate} 

  \item \case{$\statetwo = 
          \mac{\skback}{\skframe\cons\var}{\code}{\stempty}{\genv}
          \tomachasubfour
          \mac{\skback}{\skframe}{\la\var\code}{\stempty}{\genv} = \state
        $}

\begin{enumerate}
      \item \emph{Evaluating Code}: trivial since $\skphase \neq \skeval$.
      \item \emph{Stack}: trivial since the stack is empty.
      \item \emph{Frame}: any pair of the form $\skap{\codetwo}{\stack}$ or $\gcentry{\la\var\codetwo}{\stacktwo}$ in the frame $\skframe$ is also in the frame $\var\cons\skframe$, so it follows from the \ih.
      \item \emph{Environment}: any substitution of the form $\esub{\vartwo}{\codetwo}$ in the environment $\closescopem\var\cons\genv$
            is also in the environment $\genv$, so $\codetwo$ is a subterm of $\code_0$ by \ih.
      \end{enumerate} 

  \item \case{$\statetwo = 
          \mac{\skback}{ \skframe\cons\skap{\code}{\stack} }{\codetwo}{\stempty}{\genv}
          \tomachasubfive
          \mac{\skback}{\skframe}{\code\codetwo}{\stack}{\genv} = \state
        $}

\begin{enumerate}
      \item \emph{Evaluating Code}: trivial since $\skphase \neq \skeval$.
      \item \emph{Stack}: the stack $\stack$ occurs at the left-hand side in the frame $\skap{\code}{\stack}\cons\skframe$,
            so by \ih\ we know that any piece of code in $\stack$ is a subterm of $\code_0$.
      \item \emph{Frame}: any pair $\skap{\codethree}{\stack}$ or $\gcentry{\la\var\codetwo}{\stacktwo}$ in the frame $\skframe$ is also
            in the frame $\skap{\code}{\stack}\cons\skframe$, so it follows from the \ih
            
      \item \emph{Environment}: it follows from the \ih, since the environment $\genv$ is unchanged.
      \end{enumerate}

  \item \case{$\statetwo = 
          \mac{\skback}{\skframe}{\code}{\codetwo\cons\stack}{\genv}
          \tomachasubsix
          \mac{\skeval}{ \skframe\cons\skap{\code}{\stack} }{\codetwo}{\stempty} {\genv} = \state
        $}
\begin{enumerate}
      \item \emph{Evaluating Code}: note that $\codetwo$ is an element of the stack at the left-hand side of the transition,
            so by \ih\ $\codetwo$ is a subterm of $\code_0$.
      \item \emph{Stack}: trivial since the stack is empty.
      \item \emph{Frame}: any pair in the frame $\skap{\code}{\stack}\cons\skframe$ or $\gcentry{\la\var\codetwo}{\stacktwo}$ is also in the frame $\skframe$
            except for $\skap{\code}{\stack}$.
            Consider a piece of code $\codefour$ in the stack $\stack$. It is trivially also a piece of
            code in the stack $\codetwo\cons\stack$, so by the second point of the \ih\ we have that $\codefour$ is
            a subterm of $\code_0$.
      \item \emph{Environment}: it follows from the \ih, since the environment $\genv$ is unchanged.
      \end{enumerate} 
        
   \end{itemize}
  \end{itemize}

  \item \emph{Environment Size}: simply note that the only transition that extends the environment is $\tomachmtwo$.
\end{enumerate}

%% file: proofs/commutatives-and-lo-contexts.tex

The first two points rest on the following inductive mutually inductive definition of \emph{normal} and \emph{neutral} terms:

 \begin{center}
$\begin{array}{c@{\hspace{1em}}c@{\hspace{1em}}cccc}
	\AxiomC{}
	\UnaryInfC{$\var$ is neutral}
	\DisplayProof
&
	\AxiomC{$\tm$ is neutral}
	\AxiomC{$\tmtwo$ is normal}
	\BinaryInfC{$\tm \tmtwo$ is neutral}
	\DisplayProof 
	\\\\
		\AxiomC{$\tm$ is neutral}
	\UnaryInfC{$\tm$ is normal}
	\DisplayProof
	&
	\AxiomC{$\tm$ is normal}
	\UnaryInfC{$\la\var\tm$ is normal}
	\DisplayProof
\end{array}$
\end{center}

And the following two immediate properties:
\begin{enumerate}
	\item If $\relunf\code\genv$ is normal then $\code$ is normal;
	\item If $\relunf\code\genv$ is neutral then $\code$ is neutral.
\end{enumerate}

Now we can proceed with the proof.

\begin{proof}\hfill
\begin{enumerate}
\item If $\relunf\code\genv$ is normal then $\code$ is normal. By induction on the derivation of $\code$ is normal. Cases:
\begin{itemize}
\item \emph{Neutral}, \ie $\code$ is neutral. If $\relunf\code\genv$ is neutral then it follows by Point 2. Otherwise $\code = \var$ and $\relunf\code\genv$ is an abstraction, \ie $\genv = \genvtwo \cons \esublabe\var\codetwo\lababs \cons \genvthree$. Then:
$$\skamstate{\skeval}{\skframe}{\var}{\stempty}{\genv} 
 \tomachasubthree 
\skamstate{\skback}{\skframe}{\var}{\stempty}{\genv}$$

\item \emph{Abstraction}, \ie $\code = \la\var\codetwo$ with $\codetwo$ normal. Since $\relunf\code\genv = \la\var\relunf\codetwo\genv$ is normal then $\relunf\codetwo\genv$ is normal and we can use the \ih on it. Then 

\begin{center}$\begin{array}{rrcccccccc}
\skamstate{\skeval}{\skframe}{\la\var\codetwo}{\stempty}{\genv} 
& \tomachasubtwo &
\skamstate{\skeval}{\skframe\cons\var}{\codetwo}{\stempty}{\genv}  \\
 (\ih) &\tomachc^* &
\skamstate{\skback}{\skframe\cons\var}{\codetwo}{\stempty}{\genv}
& \tomachasubtwo &
\skamstate{\skback}{\skframe}{\la\var\codetwo}{\stempty}{\genv} 
\end{array}$\end{center}
\end{itemize}


\item If $\relunf\code\genv$ is neutral then $\code$ is neutral. By induction on the derivation of $\code$ is neutral. Cases:
\begin{itemize}
\item \emph{Variable}, \ie $\code = \var$. If $\relunf\code\genv$ is neutral then either $\genv(\var) = \bot$ or $\genv = \genvtwo \cons \esublabe\var\codetwo\labneu \cons \genvthree$. In both cases:
$$\skamstate{\skeval}{\skframe}{\var}{\stack}{\genv} 
 \tomachasubthree 
\skamstate{\skback}{\skframe}{\var}{\stack}{\genv}$$

\item \emph{Application}, \ie $\code = \codetwo\codethree$ with $\codetwo$ neutral and $\codethree$ normal. Since $\relunf{\codetwo\codethree}\genv$ is neutral, $\relunf\codetwo\genv$ is neutral and $\relunf\codethree\genv$ is normal, so that we can use Point 1 and the \ih on them. Then 

\begin{center}$\begin{array}{rrlccccccc}
\skamstate{\skeval}{\skframe}{\codetwo\codethree}{\stack}{\genv} 
& \tomachasubone &
\skamstate{\skeval}{\skframe}{\codetwo}{\codethree\cons\stack}{\genv}  \\
 (\ih) & \tomachc^* &
\skamstate{\skback}{\skframe}{\codetwo}{\codethree\cons\stack}{\genv}\\
& \tomachasubsix &
\skamstate{\skeval}{\skframe\cons\dentry\codetwo\stack}{\codethree}{\stempty}
{\genv}
\\
(\mbox{Point 1}) & \tomachc^* &
\skamstate{\skback}{\skframe\cons\dentry\codetwo\stack}{\codethree}{\stempty}{\genv}
& \tomachasubfive &
\skamstate{\skback}{\skframe}{\codetwo\codethree}{\stack}{\genv} 
\end{array}$\end{center}
\end{itemize}

\item If $\relunf\ctx\genv$ is a \perp context then $\ctx$ is a \perp context by \reflemmap{lo-properties}{unf-rem}. By induction on $\ctx$ being \perp (Definition \refdef{ilob-ctx}, page \pageref{def:ilob-ctx}). If $\ctx$ is empty then it is immediate. The other cases:
\begin{itemize}
\item \emph{Application Left} (rule \perpRuleAppL), \ie $\ctx = \ctxtwo \codethree$ with $\ctxtwo$ \perp and $\ctxtwo\neq\la\var\ctxthree$. Since $\relunf\ctx\genv = \relunf\ctxtwo\genv \relunf\codethree\genv$, we have that $\relunf\ctxtwo\genv$ is a \perp context and we can apply the \ih to it. Then:
$$\state =  \skamstate{\skeval}{\skframe}{\ctxtwop\codetwo \codethree}{\stack}{\genv} \tomachasubone \skamstate{\skeval}{\skframe}{\ctxtwop\codetwo}{\codethree\cons\stack}{\genv}$$

Now, if $\ctxtwo$ is empty then $\ctx$ is applicative, the stack is non-empty, and the statement is proved. Otherwise, it follows from the \ih:
$$ \skamstate{\skeval}{\skframe}{\ctxtwop\codetwo}{\codethree\cons\stack}{\genv}\underbrace{\tomachhole{\csym}^*}_{\ih} \skamstate{\skeval}{\skframetwo}{\codetwo}{\stacktwo}{\genv}
$$

\item \emph{Abstraction} (rule \perpRuleLambda), \ie $\ctx = \la\var\ctxtwo$ with $\ctxtwo$ \perp. As in the previous case, it is immediately seen that we can apply the \ih to $\ctxtwo$.

$$\state =  \skamstate{\skeval}{\skframe}{\la\var\ctxtwop\codetwo}{\stack}{\genv} \tomachasubone \skamstate{\skeval}{\skframe\cons\var}{\ctxtwop\codetwo}{\stack}{\genv} \underbrace{\tomachhole{\csym}^*}_{\ih} \skamstate{\skeval}{\skframetwo}{\codetwo}{\stacktwo}{\genv}
$$

The \emph{moreover} part follows from the \ih

\item \emph{Application Right} (rule \perpRuleAppR), \ie $\ctx = \codethree \ctxtwo$ with $\ctxtwo$ \perp and $\codethree$ neutral. Since $\relunf\ctx\genv = \relunf\codethree\genv \relunf\ctxtwo\genv$, we have that $\codethree$ is neutral (and so we can apply point 2) and $\relunf\ctxtwo\genv$ is a \perp context and so we can use the \ih on it. 

\begin{center}$\begin{array}{rrlccccccc}
\state = \skamstate{\skeval}{\skframe}{ \codethree \ctxtwop\codetwo }{\stack}{\genv} 
& \tomachasubone &
\skamstate{\skeval}{\skframe}{\codethree}{ \ctxtwop\codetwo \cons\stack}{\genv}  \\
 (\mbox{Point 2}) & \tomachc^* &
\skamstate{\skback}{\skframe}{\codethree}{ \ctxtwop\codetwo \cons\stack}{\genv}\\
& \tomachasubsix &
\skamstate{\skeval}{\skframe\cons\dentry\codethree\stack}{\ctxtwop\codetwo}{\stempty}
{\genv}

 & \underbrace{\tomachc^*}_{(\ih)} &
\skamstate{\skeval}{\skframetwo}{\codetwo}{\stacktwo}{\genv}
\end{array}$\end{center}
The \emph{moreover} part follows from the \ih

\item \emph{Erasing Redex} (rule \perpRuleGc), \ie $\ctx = (\la\var\codethree) \ctxtwo$ with $\ctxtwo$ \perp and $\var \notin \fv\codethree$. Since $\relunf\ctx\genv = (\la\var\relunf\codethree\genv) \relunf\ctxtwo\genv$ we have that $\relunf\ctxtwo\genv$ is a \perp context and so we can use the \ih on it. 

\begin{center}$\begin{array}{rrlccccccc}
\state = \skamstate{\skeval}{\skframe}{ (\la\var\codethree) \ctxtwop\codetwo }{\stack}{\genv} 
& \tomachasubone &
\skamstate{\skeval}{\skframe}{\la\var\codethree}{ \ctxtwop\codetwo \cons\stack}{\genv}  \\
 & \tomachasubseven & 
\skamstate{\skeval}{\skframe\cons \gcentry{ \la\var\codethree }\stack }{\ctxtwop\codetwo}{\stack}{\genv}\\
 & \underbrace{\tomachc^*}_{(\ih)} &
\skamstate{\skeval}{\skframetwo}{\codetwo}{\stacktwo}{\genv}
\end{array}$\end{center}

The \emph{moreover} part follows from the \ih

\end{itemize}

\end{enumerate}

\end{proof}

%% file: proofs/exponentials-decrease-proof.tex
We have $\mac{ \skeval }{ \skframe }{ \var }{ \stack }{ \genv } 
			\tomachered 
			\mac{ \skeval }{ \skframe }{ \rename{\code} }{ \stack }{ \genv }$ with $\genv(\var) = \esublabe\var\code{\labred n}$. By the labeled environment invariant (\reflemmap{umam-invariants}{env-lab}), $\genv$ is well-labeled. Then $\rename{\code} = \ctxp{\codetwo}$ with $\ctx$ \perp context. By \reflemmap{comm-and-lo-ctxs}{three} we obtain
			$$\statetwo  =  
 \skamstate{\skeval}{\skframe}{\ctxp\codetwo}{\stack}{\genv}
 \tomachc^* \skamstate{\skeval}{\skframetwo}{\codetwo}{\stacktwo}{\genv}$$
			Three cases:
\begin{enumerate}
\item $n = 1$) then by well-labeledness $\codetwo$ is a $\rtobperp$-redex $(\la\var\codethree)\codefour$, that is, either $\var \notin \fv\codethree$ or $\codefour$ is normal. First, we have
$$ \skamstate{\skeval}{\skframetwo}{(\la\var\codethree)\codefour}{\stacktwo}{\genv} \tomachasubone 
\skamstate{\skeval}{\skframetwo}{(\la\var\codethree)}{\codefour\cons\stacktwo}{\genv} \tomachm \statethree$$
Then there are 3 possible cases:
\begin{enumerate}
	\item $\codefour$ is a variable. Then $\tomachmone$ applies.
	\item $\var \in \fv\codethree$. Then $\tomachmtwo$ applies.
	\item $\var \notin \fv\codethree$ and $\codefour$ is normal and not a variable. Then:

\begin{center}$\begin{array}{rrlccccccc}
\skamstate{\skeval}{\skframetwo}{(\la\var\codethree)}{\codefour\cons\stacktwo}{\genv} \tomachm \statethree
 & \tomachasubseven & 
\skamstate{\skeval}{\skframe\cons \gcentry{ \la\var\codethree }\stack }{\codefour}{\stempty}{\genv}\\
 & \underbrace{\tomachc^*}_{(\reflemmaeqp{comm-and-lo-ctxs}{one})} &
\skamstate{\skback}{\skframe\cons \gcentry{ \la\var\codethree }\stack }{\codefour}{\stempty}{\genv}\\
 & \tomachmthree &
\skamstate{\skback}{\skframe }{\codethree}{\stack}{\genv}\\

\end{array}$\end{center}
	
\end{enumerate}

\item \emph{$n = 2$, $\genvtwo = \genvthree\cons\esublab\vartwo\codethree\cons\genvfour$ with $\lab = \lababs$, and $\ctx$ is applicative}) By well-labeledness $\codetwo$ is a variable $\vartwo$ and  $\codethree = \la\varthree\codefour$. Then:
\begin{center}$\begin{array}{rrlccccccc}
\state & \tomacheredp 2 &  
 \skamstate{\skeval}{\skframe}{\ctxp\vartwo}{\stack}{\genv}
 \underbrace{\tomachc^*}_{(\reflemmaeqp{comm-and-lo-ctxs}{three})} \skamstate{\skeval}{\skframetwo}{\vartwo}{\stacktwo}{\genv}
\end{array}$\end{center}
Always by \reflemmap{comm-and-lo-ctxs}{three}, the stack $\stacktwo$ is non-empty, say $\stacktwo = \code'\cons\stackthree$. Then we continue with:
\begin{center}$\begin{array}{rrlccccccc}
\tomacheabs &  
 \skamstate{\skeval}{\skframe}{\la\varthree\codefour}{\code'\cons\stackthree}{\genv}
\end{array}$\end{center}
Now we repeat the resoning for the case $n = 1$ completing the execution with $\tomachc^*\tomachm$.

\item \emph{$n > 1$ and $\genvtwo = \genvthree\cons\esublab\vartwo\codethree\cons\genvfour$ with $\lab = \labred {n-1}$}) By well-labeledness $\codetwo$ is a variable $\vartwo$. Then:
\begin{center}$\begin{array}{rrlccccccc}
\state & \tomacheredp n &  
 \skamstate{\skeval}{\skframe}{\ctxp\vartwo}{\stack}{\genv}
 \underbrace{\tomachc^*}_{(\reflemmaeqp{comm-and-lo-ctxs}{three})} 
 \skamstate{\skeval}{\skframetwo}{\vartwo}{\stacktwo}{\genv} 
 & \tomacheredp {n-1} &
 \skamstate{\skeval}{\skframetwo}{\codethree}{\stacktwo}{\genv} 
\end{array}$\end{center}
\end{enumerate}

%% file: proofs/bilinearity_proof.tex
 
\begin{enumerate}
 \item We use the following notion of size for stacks/frames/states:
\[\begin{array}{rcl@{\hspace{2em}}rcl}
\rmeasure{\stempty} & \defeq & 0 
&
\rmeasure{\mframe\cons\var} & \defeq & \rmeasure{\mframe}
\\

\rmeasure{\code\cons\stack} & \defeq & \size{\code} + \rmeasure{\stack} 
&
\rmeasure{\skframe\cons\skap{\code}{\stack}} & \defeq & \rmeasure{\stack} + \rmeasure{\skframe}\\

&&&
\rmeasure{\skframe\cons\gcentry{\la\var\code}{\stack}} & \defeq & \rmeasure{\code} +\rmeasure{\stack} + \rmeasure{\skframe}\\

 \rmeasure{\mac{\skeval}{\skframe}{\code}{\stack}{\genv}} & \defeq & \rmeasure{\skframe} + \rmeasure{\stack} + \size{\code} 
 &
\rmeasure{\mac{\skback}{\skframe}{\code}{\stack}{\genv}} & \defeq & \rmeasure{\skframe} + \rmeasure{\stack}
   \end{array}\]By direct inspection of the rules of the machine it can be checked that:
\begin{itemize}
\item \emph{Exponentials Increase the Size}: if $\state \tomache \statetwo$ is an exponential transition,
      then $\rmeasure{\statetwo} \leq \rmeasure{\state} + \size\tm$
      where $\size\tm$ is the size of the initial term;
      this is a consequence of the fact that exponential steps
      retrieve a piece of code from the environment, which is a subterm of the initial term by 
      \reflemmap{umam-quant-invs}{subterm};
      
\item \emph{Commutative Evaluation Transitions Decrease the Size}: if $\state \tomachhole{a} \statetwo$ with $a\in \set{\skeval\admsym_1,\skeval\admsym_2,\skeval\admsym_3,\skeval\admsym_7}$ then
      $\rmeasure{\statetwo} < \rmeasure{\state}$ (for $\skeval\admsym_3$ because the transition switches to backtracking, and thus the size of the code is no longer taken into account);
\item \emph{Multiplicative Transitions and Backtracking Transitions Decrease or do not Change the Size}: if $\state \tomachhole{a} \statetwo$ with
      $a\in \set{\monesym,\mtwosym, \mthreesym, \commfour,\commfive,\commsix}$ then $\rmeasure{\statetwo} \leq \rmeasure{\state}$.
\end{itemize}
Then a straightforward induction on $\size\exec$ shows that
\[\rmeasure{\statetwo} \leq \rmeasure{\state} + \sizee\exec \cdot \size\tm - \sizecomev\exec\]
\ie\ that $ \sizecomev\exec \leq \rmeasure{\state}  + \sizee\exec \cdot \size\tm - \rmeasure{\statetwo}$.

Now note that $\rmeasure{\cdot}$ is always non-negative and that since $\state$ is initial we have $\rmeasure{\state} = \size\tm$. We can then conclude with
\[\begin{array}{llllllllll}
   \sizecomev\exec & \leq &  \rmeasure{\state}  + \sizee\exec \cdot \size\tm - \rmeasure{\statetwo}\\
   & \leq & \rmeasure{\state}  + \sizee\exec \cdot \size\tm 
   & = & \size\tm  + \sizee\exec \cdot \size\tm & = & (1+\sizee\exec)\cdot\size\tm
  \end{array}\]

\item We have to estimate $\sizecombt\exec = \polsize\exec\commfour + \polsize\exec\commfive + \polsize\exec\commsix$. Note that
\begin{enumerate}
\item $\polsize\exec\commfour \leq \polsize\exec\commtwo$, as $\tomachasubfour$ pops variables from $\skframe$, pushed only by $\tomachasubtwo$;

\item $\polsize\exec\commfive \leq \polsize\exec\commsix$, as $\tomachasubfive$ pops pairs $\skap{\code}{\stack}$ from $\skframe$, pushed only by $\tomachasubsix$;

\item $\polsize\exec\commsix \leq \polsize\exec\commthree$, as $\tomachasubsix$ ends backtracking phases, started only by $\tomachasubthree$.
\end{enumerate}

Then $\sizecombt\exec \leq \polsize\exec\commtwo + 2\polsize\exec\commthree \leq 2\sizecomev\exec
$.

\item We have $\sizecom\exec = \sizecomev\exec + \sizecombt\exec \leq_{P. 2} \sizecomev\exec + 2\sizecomev\exec = 3\sizecomev\exec \leq_{P.1} 3\cdot(1+\sizee\exec)\cdot\size\tm $.
\end{enumerate}
